\documentclass{article}
\usepackage{booktabs} 
\usepackage[ruled]{algorithm2e} 

\SetAlFnt{\small}
\SetAlCapFnt{\small}
\SetAlCapNameFnt{\small}
\SetAlCapHSkip{0pt}
\IncMargin{-\parindent}

\usepackage{tabularx}

\usepackage[margin = 1in]{geometry}
\usepackage{bbding}
\usepackage[nointegrals]{wasysym}
\usepackage[utf8]{inputenc}
\usepackage{xspace}
\usepackage{mathtools}
\usepackage{url}
\usepackage{subcaption}
\usepackage{tikz}
\usetikzlibrary{arrows, fit}
\usetikzlibrary{backgrounds,automata,calc}
\usetikzlibrary{decorations.pathreplacing}
\usepackage{csquotes}
\usepackage{amsmath,amsfonts,amsthm}
\usepackage{pgfplots}

\DeclareMathOperator*{\argmax}{arg\,max}

\usepackage{enumerate}
\usepackage{paralist}
\usepackage[shortlabels]{enumitem}
\usepackage{appendix}
\usepackage[capitalize]{cleveref}
\usepackage{thmtools}
\usepackage{mathtools}
\usepackage{natbib}
\usepackage{changepage}

\usepackage{commath}
\newcommand{\metaunderline}[1]{#1}
\newcommand{\SUB}[2]{$\{#1, #2\}$\text{-SUB}}
\newcommand{\USW}{\texttt{USW}\xspace}

\newcommand{\MNW}{\texttt{MNW}\xspace}
\newcommand{\MMS}{\texttt{MMS}\xspace}

\newcommand{\R}{\mathbb{R}}
\newcommand{\Z}{\mathbb{Z}}

\newcommand{\Gain}{\texttt{Gain}}

\renewcommand{\cal}[1]{\mathcal{#1}}

\newcommand{\lex}{\texttt{lex}}
\newcommand{\lexmin}{\texttt{leximin}}
\newcommand{\pWel}{{p\texttt{-W}}}
\newtheorem{theorem}{Theorem}[section]
\newtheorem{lemma}[theorem]{Lemma}

\newtheorem{prop}[theorem]{Proposition}

\newtheorem*{theorem*}{Theorem}

\newcommand{\Uplambda}{\mathbin{\mathrm \Lambda}}

\theoremstyle{definition}

\theoremstyle{definition}
\newtheorem{definition}[theorem]{Definition}
\newtheorem{obs}[theorem]{Observation}
\newtheorem{example}[theorem]{Example}

\title{Dividing Good and Better Items Among Agents with Bivalued Submodular Valuations}
\author{Cyrus Cousins, Vignesh Viswanathan and Yair Zick \\
University of Massachusetts, Amherst \\
{\texttt{\{cbcousins, vviswanathan, yzick\}@umass.edu}}}
\date{}

\begin{document}

\maketitle

\begin{abstract}
We study the problem of fairly allocating a set of indivisible goods among agents with {\em bivalued submodular valuations} --- each good provides a marginal gain of either $a$ or $b$ ($a < b$) and goods have decreasing marginal gains. 
This is a natural generalization of two well-studied valuation classes --- bivalued additive valuations and binary submodular valuations. 
We present a simple sequential algorithmic framework, based on the recently introduced Yankee Swap mechanism, that can be adapted to compute a variety of solution concepts, including max Nash welfare (MNW), leximin and  $p$-mean welfare maximizing allocations when $a$ divides $b$. 
This result is complemented by an existing result on the computational intractability of MNW and leximin allocations when $a$ does not divide $b$. 
We show that MNW and leximin allocations guarantee each agent at least $\frac25$ and $\frac{a}{b+2a}$ of their maximin share, respectively, when $a$ divides $b$. 
We also show that neither the leximin nor the MNW allocation is  guaranteed to be envy free up to one good (EF1). 
This is surprising since for the simpler classes of bivalued additive valuations and binary submodular valuations, MNW allocations are known to be envy free up to \emph{any} good (EFX). 
\end{abstract}

\section{Introduction}\label{sec:intro} 

Fair allocation of indivisible items has gained significant attention in recent years. 
The problem is simple: we need to assign a set of indivisible \emph{items} to a set of \emph{agents}. 
Each agent has a subjective preference over the bundle of items they receive. 
Our objective is to find an allocation that satisfies certain \emph{fairness} and \emph{efficiency} (or more generally \emph{justice}) criteria. 
For example, some allocations maximize the product of agents' utilities, whereas others guarantee that no agent prefers another agent's bundle to their own. The fair allocation literature focuses on the existence and computation of allocations that satisfy a set of \emph{justice criteria} (see \cite{aziz2022fairdivsurvey} for a recent survey). 
For example, \citet{Plaut2017EFX} focus on computing envy free up to any good (EFX) allocations while \citet{Caragiannis2016MNW} focus on computing max Nash welfare allocations. 
Indeed, as these papers show, computing `fair' allocations without any constraint on agent valuations is an intractable problem for most justice criteria \citep{Plaut2017EFX, Caragiannis2016MNW, Kurokawa2018dichotomousallocation}.

This has led to a recent systematic attempt to ``push the computational envelope'' --- identify simpler classes of agent valuations for which there exist efficiently computable allocations satisfying multiple fairness and efficiency criteria. 
Several positive results are known when agents have \emph{binary submodular} valuations --- that is, their valuations exhibit diminishing returns to scale, and the added benefit of any item is either $0$ or $1$ \citep{Babaioff2021Dichotomous,Barman2021MRFMaxmin,Barman2021TruthfulAF,benabbou2021MRF,viswanathan2022yankee,viswanathan2022generalyankee}. 
Other works offer positive results when agents have \emph{bivalued} additive preferences, i.e., each item has a value of either $a$ or $b$, and agents' value for a bundle of items is the sum of their utility for individual items \citep{akrami2022mnw,amantidis2021mnw,ebadian2022bivaluedchores,garg2021bivaluedadditive}. 
These results are encouraging in the face of known intractability barriers. Moreover, they offer practical benefits: binary submodular valuations naturally arise in settings such as course allocation \citep{viswanathan2022yankee}, shift allocation \citep{Babaioff2021Dichotomous} and in public housing assignment \citep{benabbou2019group}. 
We take this line of work a step further and study a generalization of both bivalued additive valuations and binary submodular valuations. 

\subsection{Our Contributions}\label{sec:contrib}
When agents have binary submodular valuations,  \emph{leximin} allocations are known to satisfy several desirable criteria. 
 An allocation is leximin (or more precisely, leximin dominant) if it maximizes the welfare of the worst-off agent, then subject to that the welfare of the second worst-off agent and so on. 
 When agents have binary submodular valuations, a leximin allocation maximizes utilitarian and Nash social welfare, is envy-free up to any item (EFX), and offers each agent at least $\frac12$ of their maximin share \citep{Babaioff2021Dichotomous}. 
 Furthermore, there exists a simple sequential allocation mechanism that computes a leximin allocation \citep{viswanathan2022yankee}. 

We consider settings where agents have \emph{bivalued submodular} valuations. 
That is, the marginal contribution of any item is either $a$ or $b$, and marginal gains decrease as agents gain more items. 
We assume that $a$ divides $b$, or (by rescaling) that $a =1$ and $b$ is a positive integer. 
When agents have bivalued submodular valuations, leximin allocations no longer satisfy multiple fairness and efficiency guarantees (see \Cref{ex:simple-leximin-maxnash}). 
Thus, different justice criteria cannot be satisfied by a single allocation.
To address this, we present a general, yet surprisingly simple algorithm (called Bivalued Yankee Swap) that efficiently computes allocations for a broad number of justice criteria. 
The algorithm (\Cref{algo:bivalued-yankee-swap}) is based on the Yankee Swap protocol proposed by \citet{viswanathan2022yankee}: we start with all items unassigned. 
At every round we select an agent (based on a selection criterion $\phi$). 
If the selected agent can pick an unassigned item that offers them a high marginal benefit, they do so and we move on. Otherwise, we check whether they can \emph{steal} a high-value item from another agent. 
We allow an agent to steal a high-value item if the agent who had an item stolen from them can recover their utility by either taking an unassigned item or stealing an item from another agent. 
This results in a \emph{transfer path} where the initiating agent increases their utility by $b$, every other agent retains the same utility, and one item is removed from the pile of unassigned items. 
If no such path exists, the agent is no longer allowed to `play' for high-value items, and can either take low-value items or none at all. 

By simply modifying the agent selection criterion $\phi$, this protocol outputs allocations that satisfy a number of `acceptable' justice criteria. We think of justice criteria as ways of comparing allocations; for example, an allocation $A$ is better than an allocation $B$ according to the Nash-Welfare criterion if it either has fewer agents with zero utilities, or if the product of agent utilities under $A$ is greater than the product of agent utilities under $B$. 
A justice criterion $\Psi$ is acceptable if (informally), it satisfies a notion of Pareto dominance, and if it admits a selection criterion (also referred to as a \emph{gain function}) $\phi$ that decides which agent should receive an item in a manner consistent with $\Psi$ (see Section \ref{sec:bi-ys-sufficient} for details). 
Our main result, informally stated, is:
\begin{theorem*}[Informal]
Suppose that agents have bivalued submodular valuations. Given an acceptable justice criterion $\Psi$, there exists a simple, sequential allocation mechanism that computes an allocation maximizing $\Psi$. 
The resulting allocation lexicographically dominates any other $\Psi$-maximizing allocation. 
\end{theorem*}
Several well known justice criteria are acceptable. For every acceptable justice criterion, our result immediately implies a simple algorithmic framework that computes an allocation maximizing that justice criterion: one needs to simply implement \Cref{algo:bivalued-yankee-swap} with the appropriate gain function $\phi$. 
The key challenge is identifying an appropriate gain function to be used. 
\Cref{table:results-summary} summarizes some of the key justice criteria for which one can apply bivalued Yankee Swap, and the gain function $\phi$ required. Informally, $\phi$ takes as input an allocation $X$, an agent $i$ and an added value $d$; it measures the `improvement' to the justice criterion if agent $i$ receives an additional utility of $d$ under the allocation $X$. 
\begin{table}
\centering
\begin{tabularx}{0.8\textwidth} {>{\hsize=.8\hsize}XX}
\toprule
\textbf{Justice Criterion}&\textbf{Gain Function $\phi(X,i,d)$}\\
\toprule
 Nash Welfare (\Cref{thm:max-nash-welfare}) & $\frac{v_i(X_i) + d}{v_i(X_i)}$ \\
 \midrule
 Leximin (\Cref{thm:leximin})& $-(b+a)v_i(X_i)+d$\\
 \midrule
 $p$-mean welfare (\Cref{thm:p-mean-welfare})& $\texttt{sign}(p) [(v_i(X_i) + d)^p - v_i(X_i)^p]$ \\
 \bottomrule
\end{tabularx}
\caption{A summary of justice criteria for which bivalued Yankee Swap (\Cref{algo:bivalued-yankee-swap}) computes an optimal solution. $v_i$ and $X_i$ denote the valuation function and the allocated bundle of agent $i$ respectively.}
\label{table:results-summary}
\end{table}

To complement our algorithmic results, we further analyze Nash welfare maximizing and leximin allocations. We show that neither are guaranteed to even be envy free up to one good (EF1). 
This result shows that bivalued submodular valuations signify a departure from both binary submodular and bivalued additive valuations, where the max Nash welfare allocation is guaranteed to be envy free up to any good (EFX). 
While envy-freeness is not guaranteed under bivalued submodular valuations, we do show that max Nash welfare and leximin allocations offer approximate \MMS (maximin share) guarantees to agents. 
Specifically, we show that max Nash welfare allocations and leximin allocations are $\frac25$-\MMS and $\frac{a}{b+2a}$-\MMS respectively.


\subsection{Our Techniques}
Our theoretical analysis heavily relies on the function enumerating the number of high valued goods in an agent's bundle. 
We show (Lemma \ref{lem:beta-mrf}) that the number of high-valued goods in an agent's bundle corresponds to a binary submodular function. This observation allows us to use path augmentation techniques from the binary submodular functions literature to manipulate the number of high valued goods that each agent has. 
However, this technical observation is insufficient. 
To use path augmentation, agent bundles must not contain any low valued goods. 
To ensure this, we {\em decompose} allocations into {\em clean} allocations (a term we borrow from \citet{benabbou2019group}) and {\em supplementary} allocations. 
The clean allocation consists of all the high valued goods, and the supplementary allocation consists of all the low valued goods. 
This decomposition is easily computed when agents have bivalued additive valuations \citep{akrami2022mnw}. 
However, under bivalued submodular valuations, these decompositions need not even be unique (see Example \ref{ex:docomposition}). 
Fortunately, at least one decomposition is guaranteed to exist (see Lemma \ref{lem:decompose}).

We use this decomposition and compute path transfers solely for the clean allocation. 
Showing these path transfers are sufficient requires some casework-heavy arguments, which carefully show that bivalued Yankee Swap assigns assigns goods optimally. 

Our \MMS guarantee results utilize non-trivial (but arguably more elegant) counting arguments. 
We compare the number of high valued and low valued goods that an agent $i$ receives in the max Nash welfare allocation with the number of high and low valued goods that agent $i$ receives in a leximin allocation where all agents have the same valuation function as $i$. 
Indeed, this itself is not straightforward since there can be multiple max Nash welfare allocations, each allocating a different number of high and low valued goods to agents. 
To circumvent this, we introduce the notion of weak dominance which allows us to carefully choose an allocation where all the high valued goods are evenly distributed.

The comparison of high and low valued goods follows from counting arguments, 
that bound the difference in the number of high and low valued goods in both allocations.
These bounds are then used to obtain an \MMS guarantee for agent $i$, exploiting the crucial observation that agent $i$ obtains at least their \MMS share in the leximin allocation (where all agents have the same valuation as $i$). 
We use a very similar technique to show \MMS guarantees for leximin allocations as well. 

\subsection{Related Work}\label{sec:rel-work}
Our work is closely related to works on fair allocation under matroid rank (binary submodular) valuations. 
The problem of fair allocation under matroid rank valuations is reasonably well studied and has seen a surprising number of positive results. 
\citet{benabbou2021MRF} show that utilitarian welfare maximizing envy free up to one good (EF1) allocations always exist and can be computed efficiently.
\citet{Barman2021MRFMaxmin} show that an \MMS allocation is guaranteed to exist and can be computed efficiently. \citet{Babaioff2021Dichotomous} showed that a Lorenz dominating allocation (which is both leximin and maximizes Nash welfare) is guaranteed to exist and can be computed efficiently. 
More recently, \citet{viswanathan2022generalyankee} present a general framework (called General Yankee Swap) that can be used to compute weighted notions of fairness such as weighted max Nash welfare and weighted leximin efficiently, in addition to several others.
Almost all of the papers in this field use path transfers in their algorithms \citep{viswanathan2022yankee, viswanathan2022generalyankee, Babaioff2021Dichotomous, Barman2021MRFMaxmin}; a technique which we exploit as well.
Our algorithm (called Bivalued Yankee Swap) is closely related to the General Yankee Swap framework \citep{viswanathan2022generalyankee}. 
However, due to the complexity of bivalued submodular valuations, the analysis in our setting is more involved.

Our paper also builds on results for fair allocation under bivalued additive valuations. 
\citet{amantidis2021mnw} show that the max Nash welfare allocation is envy free up to any good (EFX) when agents have bivalued additive valuations. \citet{garg2021bivaluedadditive} showed that an EFX and Pareto optimal allocation can be computed efficiently under bivalued preferences.
\citet{akrami2022mnw} present an algorithm to compute a max Nash welfare allocation efficiently when $a$ divides $b$. 
This work is arguably the closest to ours. 
While their algorithm is different, the essential technical ingredients are the same --- their algorithm uses transfer paths and decompositions as well. 
However, their analysis is restricted to bivalued additive valuations and the max Nash welfare allocation. 
Our results generalize their work both in terms of the class of valuations and the justice criteria considered.

In a recent follow up work, \citet{akrami2022halfintegers} present a polynomial time algorithm to compute a max Nash welfare allocation for the case where $a$ divides $2b$. This case turns out to be surprisingly more complicated than the case where $a$ divides $b$, requiring a different set of techniques to manipulate the high and low valued goods.

Bivalued additive valuations have also been studied in the realm of chores. \citet{garg2022bivaluedchores} and \citet{ebadian2022bivaluedchores} present efficient algorithms to compute an EF1 and Pareto optimal allocation when agents have bivalued preferences.

\section{Preliminaries}
We use $[t]$ to denote the set $\{1, 2, \dots, t\}$. For ease of readability, we replace $A \cup \{g\}$ and $A \setminus \{g\}$ with $A + g$ and $A - g$ respectively. 

We have a set of $n$ {\em agents} $N = [n]$ and $m$ {\em goods} $G = \{g_1, g_2, \dots, g_m\}$. Each agent $i \in N$ has a {\em valuation function} $v_i : 2^G \mapsto \R_+ \cup \{0\}$; $v_i(S)$ specifies agent $i$'s value for the bundle $S$. 
Given a valuation function $v$, we let $\Delta_v(S,g) = v(S+g) - v(S)$ denote the marginal utility of the item $g$ to the bundle $S$ under $v$. When clear from context, we write $\Delta_i(S,g)$ instead of $\Delta_{v_i}(S,g)$ to denote the marginal utility of giving the item $g$ to agent $i$ given that they have already been assigned the bundle $S$.

Given $a,b \in \R_+$ such that $a < b$, we say that $v_i$ is an $(a, b)$-{\em bivalued submodular valuation} if $v_i$ satisfies the following three properties: 
\begin{inparaenum}[(a)]
\item $v_i(\emptyset) = 0$,
\item $\Delta_i(S, g) \in \{a, b\}$ for all $S \subseteq G$ and $g \in G \setminus S$, and 
\item $\Delta_i(S, g) \ge \Delta_i(T, g)$ for all $S \subseteq T \subseteq G$ and $g \in G \setminus T$.
\end{inparaenum}
We use \SUB{a}{b} to denote $(a, b)$-bivalued submodular valuation functions.

A lot of our analysis uses the class of  \SUB 01 valuations. 
\SUB01 valuations are also known as {\em binary submodular valuations}, and have been extensively studied
(see Section \ref{sec:rel-work} for a discussion). 
When $a = 0$, the valuation function $v_i$ is essentially a scaled binary submodular valuation; specifically, $\frac1b v_i(\cdot)$ is a \SUB01 valuation. 
Existing results under \SUB01 valuations trivially extend to \SUB0b valuations as well, and generally offer stronger guarantees than the ones offered in this work. Thus, we focus our attention on the case where $b > a >0$.

For ease of readability, we scale valuations by $a$. 
Under this scaling, all agents have \SUB{1}{\frac{b}{a}} valuations. 
We further simplify notation and replace the value $\frac{b}{a}$ by $c$; under this notation, all agents have \SUB{1}{c} valuations where $c > 1$. 
Note that when $a$ divides $b$, $c$ is a {\em natural number}.

An {\em allocation} $X = (X_0, X_1, X_2, \dots, X_n)$ is a partition of $G$ into $n+1$ bundles. 
Each $X_i$ denotes the bundle allocated to agent $i$; 
$X_0$ denotes the unallocated goods. 
We refer to the value $v_i(X_i)$ as the \emph{utility} of agent $i$ under the allocation $X$.
For ease of analysis, we sometimes refer to $0$ as a dummy agent with valuation function $v_0(S) = c|S|$ and allocated bundle $X_0$. This is trivially a \SUB{1}{c} function. 
None of our justice criteria consider agent $0$.

A vector $\vec y \in \R^n_{\ge 0}$ is said to {\em lexicographically dominate} another vector $\vec z \in \R^n_{\ge 0}$ if there exists some $k \in [n]$ such that for all $j \in [k - 1]$, $y_j = z_j$ and $ y_k > z_k$. This is denoted by $\vec y \succ_{\lex} \vec z$.
We define the {\em utility vector} of an allocation $X$  (denoted by $\vec u^X$) as the vector $(v_1(X_1), v_2(X_2), \dots, v_n(X_n))$. An allocation $X$ is said to {\em lexicographically dominate} another allocation $Y$ if $\vec u^X \succ_{\lex} \vec u^Y$.
 
\subsection{Justice Criteria}
We consider three central justice criteria.

\noindent\textbf{Leximin:} 
An allocation $X$ is {\em leximin} if it maximizes the least utility in the allocation and subject to that, maximizes the second lowest utility and so on. This can be formalized using the {\em sorted utility vector}. 
The sorted utility vector of an allocation $X$ (denoted by $\vec s^X$) is defined as the utility vector $(v_1(X_1), v_2(X_2), \dots, v_n(X_n))$ sorted in ascending order. 
An allocation $X$ is leximin if, for no other allocation $Y$, we have $\vec s^Y \succ_{\lex} \vec s^X$.

\noindent\textbf{Max Nash Welfare (\MNW):} Let the set of agents who receive a positive utility under an allocation $X$ be denoted by $P_X$. An allocation $X$ maximizes Nash welfare if it first maximizes the number of agents who receive a positive utility $|P_X|$ and subject to that, maximizes the value $\prod_{i \in P_X} v_i(X_i)$ \citep{Caragiannis2016MNW}.

\noindent\textbf{$p$-Mean Welfare:} 
The $p$-mean welfare of an allocation $X$ is defined as $(\frac1n \sum_{i \in N} v_i(X_i)^p)^{1/p}$ for $p \le 1$. 
Since this value is undefined when $v_i(X_i) = 0$ for any $i \in N$ and $p < 0$, we modify the definition slightly when defining a max $p$-mean welfare allocation. 
We again denote $P_X$ as the set of agents who receive a positive utility under the allocation $X$. 
An allocation $X$ is said to be a max $p$-mean welfare allocation for any $p \le 1$ if the allocation maximizes the size of $P_X$ and subject to that, maximizes $\left(\frac1n \sum_{i \in P_X} v_i(X_i)^p\right)^{1/p}$.

The $p$-mean welfare functions have been extensively studied in economics \citep{moulin2004fair}, fair machine learning \citep{heidari2018fairness,cousins2021bounds,cousins2021axiomatic,cousins2022uncertainty}, and, more recently, fair allocation \citep{barman2020tight}. When $p$ approaches $-\infty$, the $p$-mean welfare corresponds to the leximin objective and when $p$ approaches $0$, the $p$-mean welfare corresponds to the max Nash welfare objective.



Lexicographic dominance, Nash welfare and $p$-mean welfare are three ways of comparing allocations. More generally, we can think of a {\em justice criterion} $\Psi$ as a way of comparing two allocations $X$ and $Y$. Notice that we are specifically comparing the utility vectors of allocations. Therefore, we say allocation $X$ is better than $Y$ according to $\Psi$ if $\vec u^X \succ_{\Psi} \vec u^Y$. To ensure that there indeed exists a $\Psi$ maximizing allocation, we require that $\succeq_{\Psi}$ be a total ordering over all allocations. An allocation $X$ is maximal with respect to $\Psi$ if it is not $\Psi$-dominated by any other allocation, i.e., for any other allocation $Y$, it is not the case that $\vec u^Y\succ_\Psi \vec u^X$. For ease of readability, we sometimes abuse notation and use $X \succ_{\Psi} Y$ to denote $\vec u^X \succ_{\Psi} \vec u^Y$.

\section{Understanding Bivalued Submodular Valuations}
We now present some useful properties of bivalued submodular valuations that we use in our algorithms. 
Informally, we show that the number of high valued goods in any agent's bundle corresponds to a binary submodular function. This allows us to use existing techniques from the fair allocation under binary submodular functions literature in our analysis.

More formally, for each agent $i \in N+0$, we define a valuation function $\beta_i:2^G \mapsto \R^+$ defined as follows: for any $S \subseteq G$, 
$\beta_i(S)$ is equal to the size of the largest subset $T$ of $S$ such that $v_i(T) = c|T|$; 
that is, $\beta_i(S) = \max\{|T|:T \subseteq S, v_i(T) = c|T|\}$. 
In other words, all the goods in $T$ provide a marginal gain of $c$ to $i$. $\beta_i(S)$ captures the number of goods in $S$ that provide a value of $c$ to $i$. 
Our first result is that $\beta_i\in\, \SUB01$. 
To prove this, we need the following observations. 
Our first observation shows that if an agent values all of its goods at $c$, this still holds if you remove any number of goods from their bundle. 
 
\begin{obs}\label{obs:clean-subset}
For any agent $i \in N+0$, let $T$ be a set such that $v_i(T) = c|T|$. Then, for any subset $S$ of $T$, $v_i(S) = c|S|$.
\end{obs}
\begin{proof}
Each good has a marginal gain of $1$ or $c>1$. 
Therefore, the upper bound on the value of each bundle $S$ is $c|S|$. 
Removing a good from a bundle can result in a loss in value of either $1$ or $c$. Therefore, by removing $|T \setminus S|$ goods, the loss in value is upper bounded by $c|T \setminus S|$.
This gives us a lower bound of $v_i(S) \ge c|T| - c|T \setminus S| = c|S|$. Since this matches the upper bound, equality must hold. 
\end{proof}

Our second observation relates $\beta_i$ and $v_i$.
\begin{obs}\label{obs:beta-v-relation}
For each agent $i \in N+0$ and each bundle $S$, $v_i(S) = |S| - \beta_i(S) + c\beta_i(S)$.
\end{obs}
\begin{proof}
By the definition of $\beta_i(S)$, there exists some set $T \subseteq S$ of size $\beta_i(S)$ such that $v_i(T) = c|T|$. 
Let $S\setminus T = \{g_1, g_2, \dots, g_k\}$ and $Z^j = \{g_1, g_2, \dots, g_j\}$, with $Z^0 = \emptyset$. We have
\begin{align*}
    v_i(S) = v_i(T) + \sum_{j = 1}^{k} \Delta_{v_i}(T \cup Z^{j-1}, g_j).
\end{align*}
Since $\Delta_{v_i}$ is lower bounded by $1$, we have
\begin{align*}
    v_i(S) = v_i(T) + \sum_{j = 1}^{k} \Delta_{v_i}(T \cup Z^{j-1}, g_j) \ge c|T| + |S \setminus T| = c\beta_i(S) + |S| - \beta_i(S).
\end{align*}
If the inequality is strict, then it must be the case that $\Delta_{v_i}(T \cup Z^{j-1}, g_j) = c$ for some $j$. By submodularity, we must have $c \ge \Delta_{v_i}(T, g_j) \ge \Delta_{v_i}(T \cup Z^{j-1}, g_j) = c$. This implies $v_i(T+g_j) = c|T| + c = c(|T|+1)$. Since $T+g_j \subseteq S$, our choice of $T$ and therefore our assumption on $\beta_i(S)$ is contradicted. Therefore, equality must hold throughout and $v_i(S) = |S| - \beta_i(S) + c\beta_i(S)$.
\end{proof}

We are now ready to prove our main lemma.
\begin{lemma}\label{lem:beta-mrf}
For each agent $i \in N + 0$ when $v_i$ is a \SUB{1}{c} valuation, $\beta_i$ is a binary submodular valuation. 
\end{lemma}
\begin{proof}
To show this, we show that \begin{inparaenum}[(a)]
\item $\beta_i(\emptyset) = 0$, 
\item $\Delta_{\beta_i}(S,g) \in \{0,1\}$ and 
\item $\Delta_{\beta_i}(S,g) \ge \Delta_{\beta_i}(T,g)$ for all $S\subseteq T \subseteq G$ and all $g \in G\setminus T$.
\end{inparaenum}
First, $0\le \beta_i(S) \le v_i(S)$, thus  $\beta_i(\emptyset)$ is trivially $0$.
Consider some $S\subseteq G$ and some $g\in G\setminus S$. 
We need to show that $\Delta_{\beta_i}(S, g) = \beta_i(S+g) - \beta_i(S) \in \{0, 1\}$. 
Since $\beta_i$ is integer valued, $\Delta_{\beta_i}(S, g)$ must be an integer as well. Further, let $U = \argmax_{T \subseteq S: v_i(T) = c|T|} |T|$. We have the following observation by the definition of $\beta_i$.
\begin{obs}\label{obs:beta-lower-bound}
$\beta_i(S+g) \ge |U| = \beta_i(S)$.
\end{obs}

If $\beta_i(S+g) > \beta_i(S) + 1$, let $U' = \argmax_{T \subseteq S + g: v_i(T) = c|T|} |T|$. By assumption, we have $|U'| > |U| + 1$. By the definition of $\beta_i$, $U' - g$ is a subset of $S$ and $v_i(U' - g) = c|U' - g|$ (Observation \ref{obs:clean-subset}). Further, since $|U'| > |U|+1$, we have $|U' - g| > |U|$. This contradicts our assumption on $U$ and $\beta_i(S)$. Therefore, 
\begin{obs}\label{obs:beta-upper-bound}
$\beta_i(S+g) \le \beta_i(S) + 1$.
\end{obs}
Combining Observations \ref{obs:beta-lower-bound} and \ref{obs:beta-upper-bound} with the fact that $\Delta_{\beta_i}(S, g)$ is an integer, we conclude that $\Delta_{\beta_i}(S, g) \in \{0, 1\}$.

Finally, we prove that $\beta_i$ is submodular. 
Let $S, T \subseteq G$ be two sets such that $S \subseteq T$. Let $g \in G$ be a good not in $T$. We need to show that $\Delta_{\beta_i}(S, g) \ge \Delta_{\beta_i}(T, g)$. Assume for contradiction that $\Delta_{\beta_i}(S, g) < \Delta_{\beta_i}(T, g)$. Since there are only two possible values of $\Delta_{\beta_i}$ (Property (b)), we must have that $\Delta_{\beta_i}(S, g) = 0$ and $\Delta_{\beta_i}(T, g) = 1$. 

Using Observation \ref{obs:beta-v-relation}, we have the following:
\begin{align*}
    &\Delta_{v_i}(S, g) = v_i(S+g) - v_i(S) = |S+g| -\beta_i(S+g) + c \beta_i(S+g) - (|S| -\beta_i(S) + c \beta_i(S)) = 1, \\
    &\Delta_{v_i}(T, g) = v_i(T+g) - v_i(T) = |T+g| -\beta_i(T+g) + c \beta_i(T+g) - (|T| -\beta_i(T) + c \beta_i(T)) = c.
\end{align*}
The above equations use the fact that $\beta_i(S+g) = \beta_i(S)$ and $\beta_i(T+g) = \beta_i(T)+1$. 
Since $1 < c$, we have $\Delta_{v_i}(S, g) < \Delta_{v_i}(T, g)$ and $v_i$ is not submodular, a contradiction.
\end{proof}

Using $\beta_i$, we can leverage the properties of binary submodular valuations used in the fair allocation literature. However, these properties require bundles to be {\em clean} as well, as given by the following definition.

\begin{definition}[Clean]
For any agent $i \in N +0$, a bundle $S$ is {\em clean} with respect to the binary submodualar valuation $\beta_i$ if $\beta_i(S) = |S|$. By definition, this is equivalent to saying $v_i(S) = c|S|$. An allocation $X$ is said to be clean if for all agents $i \in N + 0$, $\beta_i(X_i) = |X_i|$. By our definition of $v_0$, any bundle $S$ is clean with respect to the dummy agent $0$.
\end{definition}

This is similar to the notion of non-wastefulness used in \citet{akrami2022mnw} and cleanness used in \citet{benabbou2021MRF}\footnote{The definition of clean allocations from \citet{benabbou2021MRF} is slightly different but equivalent when assuming all the agents have binary submodular valuation functions.}. 

When agents have binary submodular valuations, given a clean allocation $X$, we define the {\em exchange graph} $\cal G(X)$ as a directed graph over the set of goods. 
An edge exists from $g$ to $g'$ in the exchange graph if $g \in X_j$ and $\beta_j(X_j - g+g') = \beta_j(X_j)$ for some $j \in N+0$. 
Since our problem instances have {\em bivalued} submodular valuations, whenever we refer to the exchange graph of an allocation, we refer to it with respect to the {\em binary} submodular valuations $\{\beta_i\}_{i \in N}$. 

Let $P = (g_1, g_2, \dots, g_t)$ be a path in the exchange graph for a clean allocation $X$. 
We define a transfer of goods along the path $P$ in the allocation $X$ as the operation where $g_t$ is given to the agent who has $g_{t-1}$, $g_{t-1}$ is given to the agent who has $g_{t-2}$ and so on till finally $g_1$ is discarded. 
We call this transfer {\em path augmentation}; the bundle $X_i$ after path augmentation with the path $P$ is denoted by $X_i \Uplambda P$ and defined as $X_i \Uplambda P = (X_i - g_t) \oplus \{g_j, g_{j+1} : g_j \in X_i\}$, where $\oplus$ denotes the symmetric set difference operation. 

For any clean allocation $X$ and agent $i$, we define $F_i(X) = \{g \in G: \Delta_{\beta_i}(X_i, g) = 1\}$ as the set of goods which give agent $i$ a marginal gain of $1$ under the valuation $\beta_i$, i.e., the set of all items that give agent $i$ a marginal gain of $c$ under $v_i$. 
For any agent $i$, let $P = (g_1, \dots, g_t)$ be the shortest path from $F_i(X)$ to $X_j$ for some $j \ne i$. 
Then path augmentation with the path $P$ and giving $g_1$ to $i$ results in an allocation where $i$'s value for their bundle goes up by $c$, $j$'s value for their bundle goes down by $c$ and all the other agents do not see any change in value. 
This is formalized below and exists in \citet[Lemma 1]{Barman2021MRFMaxmin} and \citet[Lemma 3.7]{viswanathan2022yankee}.

\begin{lemma}[\citet{Barman2021MRFMaxmin}, \citet{viswanathan2022yankee}]\label{lem:path-augmentation}
    Let $X$ be a clean allocation with respect to the binary submodular valuations $\{\beta_i\}_{i \in N + 0}$. Let $P = (g_1, \dots, g_t)$ be the shortest path in the exchange graph $\cal G(X)$ from $F_i(X)$ to $X_j$ for some $i \in N+0$ and $j \in N + 0 - i$. Then, we have for all $k \in N - i - j$, $\beta_k(X_k \Uplambda P) = \beta_k(X_k)$, $\beta_i(X_i \Uplambda P + g_1) = \beta_i(X_i) + 1$ and $\beta_j(X_j \Uplambda P) = \beta_j(X_j) - 1$. Furthermore, the new allocation $X \Uplambda P$ is clean.
\end{lemma}

We now establish sufficient conditions for a path to exist. We say there is a path from some agent $i$ to some agent $j$ in an allocation $X$ if there is a path from $F_i(X)$ to $X_j$ in the exchange graph $\cal G(X)$. The following lemma appears in \citet[Theorem 3.8]{viswanathan2022yankee}.

\begin{lemma}[\citet{viswanathan2022yankee}]\label{lem:augmentation-sufficient}
Let $X$ and $Y$ be two clean allocations with respect to the binary submodular valuations $\{\beta_i\}_{i \in N + 0}$. For any agent $i$ such that $|X_i| < |Y_i|$, there is a path from $F_i(X)$ to some good in $X_k$ for some $k \in N+0$ such that $|X_k| > |Y_k|$.
\end{lemma}

Unfortunately, unlike binary submodular valuations \citep{benabbou2021MRF}, we cannot make any allocation clean without causing a loss in utility to some agents.
We can however, efficiently {\em decompose} any allocation into a clean allocation and a {\em supplementary} allocation. 
For any allocation $X$, the clean part is denoted using $X^c$ (since each good provides a value of $c$ in a clean allocation) and the supplementary allocation is denoted using $X^1$. 
The supplementary allocation is a tuple of $n$ disjoint bundles of $G$, that is $X^1 = (X^1_1, X^1_2, \dots, X^1_n)$. 

The supplementary allocation is not technically an allocation since it is not an $n+1$ partition of $G$; we call it an allocation nevertheless to improve readability. 

\begin{example}\label{ex:docomposition}
Consider an example with two agents $\{1, 2\}$ and four goods $\{g_1, \dots, g_4\}$. Agent valuations are given as follows
\begin{align*}
    v_1(S) = c|S| && v_2(S) = c\min\{|S|, 1\} + \max\{|S|-1, 0\}
\end{align*}
Consider an allocation $X$ where $X_1 = \{g_1, g_2\}$ and $X_2 = \{g_3, g_4\}$. Note that both goods in $X_1$ give agent $1$ a value of $c$ but only one good in $X_2$ gives agent $2$ a value of $c$. A natural decomposition in this case would be
\begin{align*}
    X^c_0 = \{g_4\} && X^c_1 = \{g_1, g_2\} && X^c_2 = \{g_3\} \\
    && X^1_1 = \emptyset && X^1_2 = \{g_4\}
\end{align*}
$X^c$ is the clean part and $X^1$ is the supplementary part. Note that (unlike bivalued additive valuations), decompositions need not be unique. Swapping $g_4$ and $g_3$ in the above allocation would result in another decomposition.
\end{example}
Formally, $X^c$ and $X^1$ is a decomposition of $X$ if for all agents $i \in N$, we have \begin{inparaenum}[(a)]
	\item $X^1_i, X^c_i \subseteq X_i$,
	\item $X^c_i \cap X^1_i = \emptyset$, 
	\item $X^c_i \cup X^1_i = X_i$, and 
	\item $|X^c_i| = \beta_i(X_i)$.
\end{inparaenum}
Proving that a decomposition always exists requires the following simple observation:

\begin{obs}\label{obs:beta-exchange}
Let $S, T$ be two bundles of goods which are clean with respect to some agent $i \in N+0$. If $|S| < |T|$, there exists a good $g \in T \setminus S$ such that $\Delta_{v_i}(S, g) = c$. 
\end{obs}
\begin{proof}
It is well known that if $\beta_i(S) = |S|$, $\beta_i(T) = |T|$ and $\beta_i$ is a binary submodular valuation, there exists a good $g$ in $T \setminus S$ such that $\Delta_{\beta_i}(S, g) = 1$ \citep{benabbou2021MRF}.

Combining this result with Observation \ref{obs:beta-v-relation}, we have that $\Delta_{v_i}(S, g)$ equals
\begin{align*}
	v_i(S+g) - v_i(S) &= |S+g| -\beta_i(S+g) + c \beta_i(S+g) - (|S| -\beta_i(S) + c \beta_i(S)) = 1 + \Delta_{\beta_i}(S, g) (c-1),
\end{align*}
which equals $c$.
\end{proof}
Using \Cref{obs:beta-exchange}, we show that a decomposition always exists.
\begin{lemma}\label{lem:decompose}
For any allocation $X$, there exists a decomposition into a clean allocation $X^c$ and a supplementary allocation $X^1$ such that for all $i \in N$, we have 
\begin{inparaenum}[(a)]
	\item $X^1_i, X^c_i \subseteq X_i$,
	\item $X^c_i \cap X^1_i = \emptyset$, 
	\item $X^c_i \cup X^1_i = X_i$, and 
	\item $|X^c_i| = \beta_i(X_i)$.
\end{inparaenum}
\end{lemma}
\begin{proof}
We construct the decomposition as follows: initialize $X^c$ with an empty allocation. Then, for each agent $i \in N$, add goods one by one from $X_i$ to $X^c_i$ if the good provides a marginal value of $c$. Stop when no good in $X_i \setminus X^c_i$ provides a marginal value of $c$ to any agent $i \in N$. For each agent $i \in N$, set $X^1_i = X_i \setminus X^c_i$. Finally, set $X^c_0 = G \setminus \bigcup_{i \in N} X^c_i$.

Note that this construction ensures that $X^c$ is clean and trivially satisfies (a), (b) and (c). It remains to show that $|X^c_i| = \beta_i(X_i)$ for each $i \in N$.

Recall that $\beta_i(X_i)$ is defined as the size of the largest subset $T$ of $X_i$ (need not be unique) such that each good in $T$ is valued at $c$. Via the construction of $X^c_i$, we have $|X^c_i| \le \beta_i(X_i) = |T|$. If the inequality is strict, using Observation \ref{obs:beta-exchange}, we get that there exists a good $g$ in $T \setminus X^c_i$ such that $\Delta_{v_i}(X^c_i, g) = c$ contradicting our construction of $X^c$. Therefore, equality holds and (d) is true. 
\end{proof}

Note that in the decomposition above, $X^c_0$ contains all the goods in $X^1$. This is done to ensure $X^c$ is an $n+1$ partition of $G$. We do not need to do this for $X^1$. 

We define the union of a clean allocation and a supplementary allocation $X^c \cup X^1$ as follows: for each agent $i \in N$, $X_i = X^c_i \cup X^1_i$ and $X_0 = G \setminus \bigcup_{i \in N} X_i$. This definition holds for any pair of clean and supplementary allocations and not just decompositions via Lemma \ref{lem:decompose}. It is easy to see that if an allocation $X$ was decomposed into $X^c$ and $X^1$ (satisfying the properties of Lemma \ref{lem:decompose}), then $X = X^c \cup X^1$. 
We will refer to an allocation $X$ as $X = X^c \cup X^1$ to denote a decomposition of the allocation (via Lemma \ref{lem:decompose}) into a clean allocation $X^c$ and a supplementary allocation $X^1$. As we saw in Example \ref{ex:docomposition}, decompositions need not be unique. When we use $X = X^c \cup X^1$, we will refer to any one of the possible decompositions. The exact one will not matter.

Finally, we present a useful metric to compare allocations using their decompositions. 
We refer to this metric as {\em domination}. 
We first compare the sorted utility vectors of $X^c$ and $Y^c$. If the sorted utility vector $X^c$ lexicographically dominates that of $Y^c$, then $X$ dominates $Y$; if the two allocations $X^c$ and $Y^c$ have the same sorted utility vectors, we compare their utility vectors. 
If $X^c$ is lexicographically greater than $Y^c$, then $X$ dominates $Y$. If $X^c$ and $Y^c$ have the same utility vectors, we compare $X$ and $Y$ lexicographically.  
This definition is formalized below. 
Recall that $\vec u^X$ and $\vec s^X$ denote the utility vector and sorted utility vector of the allocation $X$ respectively.
\begin{definition}[Domination]
    We say an allocation $X = X^c \cup X^1$ {\em dominates} an allocation $Y = Y^c \cup Y^1$ if any of the following three conditions hold:
    \begin{enumerate}[(a)]
        \item $\vec s^{X^c} \succ_{\lex} \vec s^{Y^c}$,
        \item $\vec s^{X^c} = \vec s^{Y^c}$ and $\vec u^{X^c} \succ_{\lex} \vec u^{Y^c}$, or
        \item $\vec u^{X^c} = \vec u^{Y^c}$ and $\vec u^X \succ_{\lex} \vec u^Y$.
    \end{enumerate}
    An allocation $X$ is a {\em dominating} $\Psi$ maximizing allocation if no other $\Psi$ maximizing allocation $Y$ dominates $X$.
\end{definition}

\section{Bivalued Yankee Swap}
We now present {\em Bivalued Yankee Swap} --- a flexible framework for fair allocation with bivalued $(1,c)$ submodular valuations. The results in this section assume that $c$ is a natural number; in other words, we are interested in $(a,b)$ bivalued submodular valuations where $a$ divides $b$.

In the original Yankee Swap algorithm \citep{viswanathan2022yankee}, all goods start off initially unallocated. 
The algorithm proceeds in rounds; at each round, we select an agent based on some criteria. 
This agent can either take an unallocated good, or initiate a transfer path by stealing a good from another agent, who then steals a good from another agent and so on until some agent steals a good from the pool of unallocated goods. 
\citet{viswanathan2022yankee} show that these transfer paths can be efficiently computed and are equivalent to paths on the exchange graph. 
If there is no transfer path from the agent to an unassigned item, the agent is removed from the game. We continue until either all items have been assigned.

The Bivalued Yankee Swap algorithm is a modified version of this approach. 
We start by letting agents run Yankee Swap, but require that whenever an agent receives an item (by either taking an unassigned item or by stealing an item from another agent), that item must offer them a marginal gain of $c$. 
In addition, we require that every agent who had an item stolen from them fully recovers their utility, i.e., an agent who lost an item of marginal value $c$, must receive an item of marginal value $c$ in exchange; an agent who lost an item of marginal value $1$ must receive an item of marginal value $1$ in exchange.
Thus, whenever an agent initiates a transfer path, that path results in them receiving an additional utility of $c$, while all other agents' utilities remain the same. 
More formally, every agent starts in the set $U$. 
If the agent is able to find a path in the exchange graph to an unallocated good, we augment the allocation using this path. 
If the agent is unable find a path, we remove them from $U$. 
Once the agent is removed from $U$, no transfer path exists that can give the agent a value of $c$. 
Therefore, if an agent outside of $U$ is chosen, we {\em provisionally} give the agent an arbitrary unassigned good, offering them a marginal gain of $1$. 
Since all unassigned goods offer a marginal gain of $1$ for every agent not in $U$, the choice of which good to allocate to guarantee a value of $1$ does not matter. 
Therefore, 
we treat this provisionally allocated good as an unallocated good as well; thereby allowing transfer paths to steal this good. 
If this provisionally allocated good gets stolen, we replace it with another low-value good. 
The algorithm stops when there are no unallocated goods left. 
These steps are described in Algorithm \ref{algo:bivalued-yankee-swap}.

\begin{algorithm}[t]
    \LinesNumbered
    \SetAlgoVlined
    \caption{Bivalued Yankee Swap}
    \label{algo:bivalued-yankee-swap}
    \DontPrintSemicolon
    \SetKwInOut{Input}{Input}
    \SetKwInOut{Output}{Output}
    \Input{A set of goods $G$, a set of agents $N$ with $\{1, c\}$-SUB valuations $\{v_h\}_{h \in N}$ and a gain function $\phi$.}
    \Output{A dominating $\Psi$ maximizing allocation.}
    $X^c = (X^c_0, X^c_1, \dots, X^c_n) \gets (G, \emptyset, \dots, \emptyset)$ \tcc*{$X^c$ has no goods allocated.}
    \tcp{Invariant: $X^c_0$ stores the provisionally allocated goods as well as the unallocated goods.}
    $X^1 = (X^1_1, \dots, X^1_n) \gets (\emptyset, \dots, \emptyset)$ \tcc*{Stores the provisional allocation.}
    $U\gets N$ \tcc*{Set of agents still in play for $c$ valued goods.}
    \While{$\sum_{h \in N}|X^1_h| < |X^c_0|$}{
    \tcc*{While unallocated goods exist.}
        $\Gain_{c} = 
            \begin{cases}
                \max_{k \in U} \phi(X^c \cup X^1, k, c) & U \ne \emptyset \\
                -\infty & U = \emptyset
            \end{cases}
        $\;
        $\Gain_{1} = 
            \begin{cases}
                \max_{k \in N\setminus U} \phi(X^c \cup X^1, k, 1) & N \setminus U \ne \emptyset \\
                -\infty & N \setminus U = \emptyset
            \end{cases}
        $\;
        \uIf{\normalfont $\Gain_c \ge \Gain_1$}{ 
        \tcp{Try to give an agent from $U$ a value of $c$.}
        $S \gets \argmax\limits_{k \in U} \phi(X^c \cup X^1, k, c)$\;
        \vspace*{-0.11cm}
        $i \gets \min_{j \in S} j$\;
        \eIf{a path in the exchange graph $\cal G(X^c)$ from $F_i(X^c)$ to $X^c_0$ exists}{
            $P = (g'_{1}, g'_{2}, \dots, g'_{k}) \gets$ the shortest path from $F_i(X^c)$ to $X^0$ in $\cal G(X^c)$\;
            \tcp{Augment the allocation with the path $P$ and give $g'_{1}$ to $i$}
            $X^c_k \gets X^c_k \Uplambda P$ for all $k \in N + 0 - i$\;
            $X^c_i \gets X^c_i \Uplambda P + g'_{1}$\;
            \If{$g'_{k} \in X^1_j$ for some $j \in N$}{
                \tcp{Replace good stolen from $X^1_j$ with an arbitrary unallocated good}
                $X^1_j \gets X^1_j - g'_{k} + g$ for some $g \in X^c_0 \setminus \bigcup_{h \in N} X^1_h$\;
            }
        }
        {
            $U\gets U - i$ \tcc*{$i$ is no longer in play for $c$ valued goods.}
        }
        }
        \ElseIf{\normalfont $\Gain_c < \Gain_1$ }{ 
        \tcp{Give an agent from $N \setminus U$ a value of $1$.}
            $S \gets \argmax\limits_{k \in N \setminus U} \phi(X^c \cup X^1, k, 1)$\;
            $i \gets \min_{j \in S} j$\;
            $X^1_i \gets X^1_i + g$ for some $g \in X^c_0 \setminus \bigcup_{h \in N} X^1_h$\;
    }
    }
    \Return $X^c \cup X^1$\;
\end{algorithm}

We use a {\em gain function} $\phi$ to pick the next agent to invoke a transfer path. 
Informally, $\phi$ computes the change in the justice criterion $\Psi$ when a good is assigned to an agent. 
The gain function concept is used to compute allocations satisfying a diverse set of justice criteria in \citet{viswanathan2022generalyankee}, when agents have binary submodular valuations. 

More formally, $\phi$ maps a tuple $(\vec u^X, i, d)$ consisting of the utility vector of an allocation $X$, an agent $i$ and a value $d \in \{1, c\}$ to a finite real value. 
The value $\phi(\vec u^X, i, d)$ quantifies the gain in the justice criterion $\Psi$ of adding a value of $d$ to the agent $i$ given the allocation $X$. 
We abuse notation and use $\phi(X, i, d)$ to denote $\phi(\vec u^X, i, d)$.

At each round, we find the agent $i \in U$ who maximizes $\phi(X, i, c)$ (implicitly assuming that $i$ can find an item with a marginal gain of $c$), and the agent $j \in N\setminus U$ who maximizes $\phi(X, j, 1)$ (implicitly assuming that all items offer $j$ a marginal gain of $1$).  
We then choose the agent, among $i$ and $j$, who maximizes $\phi$; in the algorithm, these values are denoted by $\Gain_c$ and $\Gain_1$. Ties are broken first in favor of agents in $U$ and second, in favor of agents with a lower index. 
Let us present an example of how \Cref{algo:bivalued-yankee-swap} works, with two candidate gain functions: \Cref{ex:simple-leximin-maxnash} shows how a leximin allocation and an \MNW allocation are computed for a simple two-agent instance.
\begin{example}\label{ex:simple-leximin-maxnash}
Consider a setting with six items $G = \{g_1,\dots,g_6\}$, and two additive agents $1$ and $2$. Agent $1$'s valuation function is $v_1(S) = |S|$, and Agent 2's valuation is $v_2(S) = 5|S|$. In other words, Agent $1$ values every item at $1$, whereas Agent $2$ values every item at $5$. We set the gain function to be $\phi_{\lexmin}(X,i,d) = -6v_i(X_i)+d$, to compute a leximin allocation (as per Theorem \ref{thm:leximin}), and run \Cref{algo:bivalued-yankee-swap}. All agents are initially in $U$, and we compute $\Gain_c,\Gain_1$. 
In the first iteration, $\Gain_1 = - \infty$ since all agents are in $U$ and $\Gain_c = - 6 \times v_1(\emptyset) + 5 = 5$. Both agents $1$ and $2$ have equal $\phi$ values so we break ties and choose Agent $1$.
However, there is no way to give Agent $1$ a good which gives them a marginal gain of $5$. We therefore remove Agent $1$ from $U$.

In the next iteration, $\Gain_c = -6\times v_2(\emptyset)+5 = -6\times 0 +5 = 5$, and $\Gain_1 = -6\times v_1(\emptyset)+1 = 1$.
Since $\Gain_c > \Gain_1$, we choose Agent $2$ to receive an item. Since they value all the items at $5$, we pick an arbitrary good (say $g_1$) and allocate it to them.

In the next iteration, $\Gain_c = -6\times v_2(g_1)+5 = -6\times 5 +5 = -25$ but $\Gain_1 = -6\times v_1(\emptyset)+1 = 1$ still. So we have $\Gain_1 > \Gain_c$ and we choose agent $1$ to provisionally receive an arbitrary good (say $g_2$). 
We will have $\Gain_1 > \Gain_c$ for the remaining iterations as well, yielding the allocation $X_1 = \{g_2,\dots,g_6\}, X_2 = \{g_1\}$, which is indeed leximin.   

By modifying the gain function $\phi$ to be 
$$\phi_{\MNW}(X,i,d) = \begin{cases} \frac{v_i(X_i)+d}{v_i(X_i)} &v_i(X_i) > 0\\ Md & v_i(X_i) =0\end{cases},$$ as per Theorem \ref{thm:max-nash-welfare} (where $M$ is a very large number), \Cref{algo:bivalued-yankee-swap} outputs an \MNW allocation. 
When no items are assigned, the first iteration proceeds the exact same way as that of the leximin allocation and Agent $1$ is removed from $U$. 

In the second iteration, we have $\Gain_c = 5\times M$ and $\Gain_1 = M$.
Thus, Agent 2 receives an item (say $g_1$). Next, we have that $\Gain_c = \phi(X,2,c) = \frac{v_2(g_1) + c}{v_2(g_1)} = \frac{2c}{c} = 2$, and $\Gain_1 = \phi(X,1,1) = M\times 1=M$. Now we have $\Gain_1 > \Gain_c$, so Agent 1 gets $g_2$. 
In the next iteration, we still have $\Gain_c = \phi(X,2,c) = \frac{2c}{c} = 2$, but $\Gain_1 = \phi(X,1,1) = \frac{v_1(g_2)+1}{v_1(g_2)} = \frac21 = 2$ as well. Thus, according to our tiebreaking scheme, we let agent $2$ pick an item next, and they get $g_3$. We continue in a similar manner, and end up with the allocation $X_1 = \{g_2,g_4,g_6\},X_2 = \{g_1,g_3,g_5\}$, which is indeed \MNW. 
\end{example}


\subsection{When Does Bivalued Yankee Swap Work?}\label{sec:bi-ys-sufficient}
Bivalued Yankee Swap computes a $\Psi$-maximizing allocation when the following conditions are satisfied. The conditions are defined for any general vector $\vec x$ but it would help to think of $\vec x, \vec y$ and $\vec z$ as utility vectors.

\begin{description}
	\item[(C1) -- Symmetric Pareto Dominance] 
	Let $\vec x, \vec y \in \Z^n_{\ge 0}$ be two vectors. Let $s(\vec x)$ denote the vector $\vec x$ sorted in increasing order.
	If for all $i \in N$, $s(\vec x)_i \ge s(\vec y)_i$, then $\vec x \succeq_{\Psi} \vec y$. 
	Equality holds if and only if $s(\vec x) = s(\vec y)$.
 \item[(C2) -- Gain Function] $\Psi$ admits a finite valued gain function $\phi$ that satisfies the following properties:
	\begin{description}
		\item[(G1)] Let  $\vec x \in \Z^n_{\ge 0}$ be some utility vector, let $i, j \in N$ be two agents, and $d_1, d_2$ be two values in $\{1, c\}$.  
		Let $\vec y \in \Z^n_{\ge 0}$ be the vector resulting from adding a value of $d_1$ to $x_i$. 
		Let $\vec z$ be the vector resulting from adding a value of $d_2$ to $x_j$. 
		If $\phi(\vec x, i, d_1) \ge \phi(\vec x, j, d_2)$, we must have $\vec y \succeq_{\Psi} \vec z$. 
		Equality holds if and only if $\phi(\vec x, i, d_1) = \phi(\vec x, j, d_2)$.
		\item[(G2)] For any two vectors $\vec x, \vec y \in \Z^n_{\ge 0}$, an agent $i \in N$ such that $ x_i \le  y_i$ and any $d \in \{1, c\}$, we must have $\phi(\vec x, i, d) \ge \phi(\vec y, i, d)$. 
		Equality holds if $x_i = y_i$.
		\item[(G3)] 
		For any vector $\vec x \in \Z^n_{\ge 0}$ and any two agents $i, j \in N$, if $x_i \le x_j$, then $\phi(\vec x, i, d) \ge \phi(\vec x, j, d)$ for any $d \in \{1, c\}$. 
		Equality holds if and only if $x_i = x_j$.
	\end{description}
\end{description}
There are two differences between our conditions and the conditions of the General Yankee Swap algorithm \citep{viswanathan2022generalyankee}. 
First, we strengthen Pareto Dominance to Symmetric Pareto Dominance (C1). 
Symmetric Pareto Dominance is not biased towards any agent and therefore, two allocations with the same sorted utility vector have the same objective value. 
As an immediate corollary, weighted notions of fairness like the max weighted Nash welfare objective \citep{chakraborty2021weighted} do not satisfy Symmetric Pareto Dominance (C1) --- when agents have weights, two allocations with the same sorted utility vector may not have the same objective value. 
Second, we introduce (G3) which further strengthens our conditions --- (G3) states that all things being equal, it is better to increase the utility of lower utility agents than higher utility agents. Our main result is the following.

\begin{restatable}{theorem}{thmbivaluedyankeeswap}\label{thm:bivalued-yankee-swap}
Suppose that all agents in $N$ have \SUB 1c valuations where $c$ is a positive integer greater than or equal to $2$. 
Let $\Psi$ be a justice criterion that satisfies (C1) and (C2), with a gain function $\phi$. 
Bivalued Yankee Swap with the gain function $\phi$ outputs a dominating $\Psi$-maximizing allocation.
\end{restatable}

\section{Proof of Theorem \ref{thm:bivalued-yankee-swap}}\label{apdx:bivalued-yankee-swap}

Before we show our main result, we present some useful lemmas. 
These lemmas are categorized into subsections. 

\subsection{Properties of $\Psi$}



In this section, we show that adding a value of $c$ to an agent is always strictly better than adding a value of $1$ to an agent. 

\begin{lemma}\label{lem:g4}
For any vector $\vec x \in \Z^n_{\ge 0}$ and any agent $i \in N$, $\phi(\vec x, i, c) > \phi(\vec x, i, 1)$.
\end{lemma}
\begin{proof}
Let $\vec y \in \Z^n_{\ge 0}$ be the vector starting at $\vec x$ and adding a value of $c$ to $x_i$.
Similarly, let $\vec z \in \Z^n_{\ge 0}$ be the vector starting at $\vec x$ and adding a value of $1$ to $x_i$. 

For all $h \in N$, $y_h \ge z_h$ with a strict inequality for agent $i$. 
Therefore, $\vec y \succ_{\Psi} \vec z$ by symmetric Pareto dominance (C1). 
Using (G1), this gives us $\phi(\vec x, i, c) > \phi(\vec x, i, 1)$.
\end{proof}

We also simplify (G2), which makes our proofs slightly easier.
\begin{lemma}\label{lem:g5}
Let $\vec x \in \Z^n_{\ge 0}$ be some vector. Define $\vec y \in \Z^n_{\ge 0}$ as follows
\begin{align*}
    y_h = 
    \begin{cases}
        x_h & h \in N - i - j \\
        x_i - d_i & h = i \\
        x_j + d_j & h = j \\
    \end{cases}
\end{align*}
for some $i, j \in N$ and $d_i, d_j \in \{1, c\}$. Then $\phi(\vec x, j, d_j) \ge \phi(\vec y, i, d_i)$ if and only if $\vec y \succeq_{\Psi} \vec x$ with equality holding if and only if $\phi(\vec x, j, d_j) = \phi(\vec y, i, d_i)$.
\end{lemma}
\begin{proof}
Let $\vec z$ be a vector defined as follows:
\begin{align*}
    z_h = 
    \begin{cases}
        x_h & h \in N - i - j \\
        x_i - d_i & h = i \\
        x_j & h = j \\
    \end{cases}
    .
\end{align*}
Using (G2) and the fact that $x_j = z_j$, $\phi(\vec x, j, d_j) = \phi(\vec z, j, d_j)$; similarly, since $y_i = z_i$, $\phi(\vec y, i, d_i) = \phi(\vec z, i, d_i)$.
Note that $\vec y$ results from offering $j$ a gain of $d_j$ under $\vec z$, and $\vec x$ results from offering $i$ a gain of $d_i$ under $\vec z$. Thus, using (G1), $\phi(\vec z, j, d_j) \ge \phi(\vec z, i, d_i)$ if and only if $\vec y \succeq_{\Psi} \vec x$, with equality holding if and only if $\phi(\vec z, j, d_j) = \phi(\vec z, i, d_i)$. 
However, $\phi(\vec z,j,d_j) = \phi(\vec x,j,d_j)$, and $\phi(\vec z,i,d_i) =\phi(\vec y,i,d_i)$. Substituting these equations into the condition we state on $\vec z$ via (G1) yields the required result.
\end{proof}
\subsection{Invariants in Bivalued Yankee Swap}\label{sec:invariants-biYS}
In this section, we show some important invariants in Bivalued Yankee Swap.
Our first Lemma shows that the allocation $X^c$ is always clean.

\begin{lemma}\label{lem:xc-clean}
	At every iteration of Bivalued Yankee Swap, the allocation $X^c$ is clean.
\end{lemma}
\begin{proof}
Given our valuation of the agent $0$, we start at a clean allocation. The only way we update $X^c$ is via path augmentation on the exchange graph. This operation results in a clean allocation (Lemma \ref{lem:path-augmentation}). Therefore, the lemma holds.
\end{proof}


Our next Lemma proves an important relation between the utility of agents in the allocation $X^c \cup X^1$ in any iteration of the Bivalued Yankee Swap algorithm.
\begin{lemma}\label{lem:bivalued-picks-least-value}
    At the start of any iteration of the Bivalued Yankee Swap algorithm, let the values of $X^1$ and $X^c$ be $W^1$ and $W^c$ respectively. Let $W = W^c \cup W^1$. The following holds:
    \begin{enumerate}[(a)]
        \item If the chosen agent $i$ is in $U$, then for any other agent $j \in U$, we have $v_i(W_i) \le v_j(W_j)$. If $v_i(W_i) = v_j(W_j)$, then $i < j$.
        \item If the chosen agent $i$ is not in $U$, then for any other agent $j \notin U$, we have $v_i(W_i) \le v_j(W_j)$. If $v_i(W_i) = v_j(W_j)$, then $i < j$.
    \end{enumerate}
\end{lemma}
\begin{proof}
We show (a) here; (b) can be shown similarly. 
Assume for contradiction that there exists $j \in U$ such that $v_i(W_i) > v_j(W_j)$. From (G3), we get that $\phi(W, i, c) < \phi(W, j, c)$ --- this contradicts our choice of $i$.
  
Now assume $v_i(W_i) = v_j(W_j)$. Using (G3), we must have $\phi(W, i, c) = \phi(W, j, c)$. Ties are broken by choosing the agent with the lower index; therefore $i < j$. 
\end{proof}

Our next Lemma shows that the goods in $X^1$ are low value goods. That is, they always provide a marginal value of $1$ to the agents they are allocated to.

\begin{prop}\label{prop:x-reduction}
At any iteration in the Bivalued Yankee Swap algorithm, if $X^c$ and $X^1$ are the allocations maintained by the algorithm at the start of this iteration, then $v_i(X_i) = c|X^c_i| + |X^1_i|$ where $X = X^c \cup X^1$.
\end{prop}
\begin{proof}
From Lemma \ref{lem:xc-clean}, we get that $X^c$ is a clean allocation. Therefore, $v_i(X_i) \ge c|X^c_i| + |X^1_i|$. If the inequality is strict, then it must be the case that for some good $g \in X^1_i$, we have $\Delta_i(X^c_i, g) = c$.

Let us now consider the iteration (denoted by $t$) where $i$ was removed from $U$. Let $W^c$ and $W^1$ be the allocation at the start of this iteration. We must have that $X^c_i = W^c_i$ by the definition of the iteration we chose. Further, note that $g$ is an unallocated good at this iteration since unallocated goods, once allocated do not become unallocated again and we know that $g$ was allocated to $i$ at some iteration $t' > t$. If $\Delta_i(W^c_i, g) = c$, and $g$ is unallocated, we have a trivial path from $F_i(W^c)$ to the pool of unallocated goods in the exchange graph $\cal G(W^c)$. This implies that we cannot have $\Delta_i(X^c_i, g) = c$. Therefore, the inequality cannot be strict and we must have $v_i(X_i) = c|X^c_i| + |X^1_i|$. 
\end{proof}

\subsection{Properties of the Output of Bivalued Yankee Swap}\label{sec:properties-biYS}
Let us next examine the allocation $X = X^c \cup X^1$ output by Bivalued Yankee Swap; here again, $X_i^c$ are items whose marginal value for $i$ is $c$, and $X_i^1$ is the set of items whose marginal value to $i$ is $1$.
Our first lemma shows that the algorithm terminates in polynomial time when the gain function $\phi$ can be computed efficiently.

\begin{restatable}{prop}{propbivaluedyankeetime}\label{prop:bivalued-yankee-time}
The Bivalued Yankee Swap algorithm runs in $O([m^2(n + \tau) + nT_\phi(n, m)](m+n))$ time where $\tau$ upper bounds the complexity of computing the value of any bundle of goods for any agent and $T_\phi(n, m)$ is the complexity of computing $\phi$.
\end{restatable}
\begin{proof}
Our proof uses the arguments of \citet{viswanathan2022yankee} who show that the Yankee Swap algorithm runs in $O(m^2(n + \tau)(m+n))$ time. We assume both $X^c$ and $X^1$ are stored as binary matrices, and that adding/removing a good takes $O(1)$ time.

At each iteration of the algorithm, computing $\Gain_1$ and $\Gain_c$ takes $n\times T_{\phi}(n, m)$ time. 
\begin{enumerate}[label={\bfseries Case \arabic*:},itemindent=*,leftmargin=0cm]
\item $\Gain_1 > \Gain_c$. In this case, finding the agent $i$ takes $n T_{\phi}(n, m)$ time. Finding the good $g \in X^c_0 \setminus (\bigcup_{j \in N} X^1_j)$ and giving it to agent $i$ takes at most $O(nm)$ time. This gives us the complexity of Case 1 as $O(nT_{\phi}(n, m) + nm)$.
\item $\Gain_1 \le \Gain_c$. In this case, finding the agent $i$ still takes $n T_{\phi}(n, m)$ time.
\end{enumerate}
We make three observations. 
Constructing the exchange graph, checking for a path and transferring goods along the path can be done in $O(m^2(n + \tau))$ time. This is shown by \citet[Section 3.4]{viswanathan2022yankee} when $X^c$ is stored as a binary matrix.
Finally, checking if a good was stolen from $X^1_j$ and finding a good $g$ to replace the stolen good with can be done trivially in $O(nm)$ time: check each possible $j \in N$ for a stolen good, then find $g$ to replace the stolen good with. This gives us the complexity of Case 2 as $O(nT_{\phi}(n, m) + m^2(n+\tau))$.

We have the following observation:
\begin{obs}\label{obs:each-iteration-complexity}
   The time complexity of each iteration is $O(nT_{\phi}(n, m) + m^2(n+\tau))$.
\end{obs}

Let us now bound the number of iterations of Bivalued Yankee Swap. At each iteration, one of the following three events occur:
\begin{enumerate}[(a)]
    \item $|X^c_0|$ decreases by 1. $|\bigcup_{j \in N} X^1_j|$ and $U$ remain unchanged. This occurs when $\Gain_c \ge \Gain_1$ and a path exists.
    \item $|\bigcup_{j \in N} X^1_j|$ increases by 1. $|X^c_0|$ and $U$ remain unchanged. This occurs when $\Gain_c < \Gain_1$.
    \item $|U|$ decreases by 1. $|\bigcup_{j \in N} X^1_j|$ and $|X^c_0|$ remain unchanged. This occurs when $\Gain_c \ge \Gain_1$ and a path does not exist.
\end{enumerate}

Summarizing the above three points, at each iteration, $|X^c_0| - |\bigcup_{j \in N} X^1_j|$ decreases by $1$ and $|U|$ remains constant or $|X^c_0| - |\bigcup_{j \in N} X^1_j|$ remains constant and $|U|$ decreases by $1$.
Since the initial value of $|U|$ and $|X^c_0| - |\bigcup_{j \in N} X^1_j|$ are $n$ and $m$ respectively, the total number of iterations is upper bounded at $m+n$. Combining this with Observation \ref{obs:each-iteration-complexity} gives us the required result.
\end{proof}

Our next lemma shows that agents in $U$ at the end of the algorithm always have a weakly greater utility than agents outside $U$.
\begin{lemma}\label{lem:not-in-p-less-than-in-p}
Let $X = X^c \cup X^1$ be the allocation output by Bivalued Yankee Swap. If $i \in U$ at the end of the algorithm and $j \notin U$, $v_i(X_i) \ge v_j(X_j)$.  If equality holds, $|X^c_i| > |X^c_j|$ or $|X^c_i| = |X^c_j|$ and $j < i$.
\end{lemma}
\begin{proof}
We divide our proof into two cases:
\begin{enumerate}[label={\bfseries Case \arabic*:},itemindent=*,leftmargin=0cm]
\item $|X^1_j| > 0$.
Assume for contradiction that $v_i(X_i) < v_j(X_j)$. 
Since all bundles are integer valued, we must have $v_j(X_j) - 1 \ge v_i(X_i)$. Let $\hat X$ be the allocation at the start of the iteration where agent $j$ received their final (provisionally allocated) good. We have $v_j(\hat{X}_j) = v_j(X_j) - 1 \ge v_i(X_i) \ge v_i(\hat{X}_i)$. We therefore have $\phi(\hat{X}, j, 1) < \phi(\hat{X}, j, c) \le \phi(\hat{X}, i, c)$; the first inequality holds due to Lemma \ref{lem:g4}, and the second inequality holds due to (G3). This is a contradiction since $i \in U$, therefore $i$ should have been chosen instead of $j$ at the iteration in discussion.

If $v_i(X_i) = v_j(X_j)$ and $|X^1_j| > 0$, we have $|X^c_j| < |X^c_i|$. 

\item $|X^1_j| = 0$.
If $|X^1_j| = 0$, consider the iteration where $j$ was removed from $U$. Let $\hat{X}$ be the allocation at this iteration. From Lemma \ref{lem:bivalued-picks-least-value}, we have that $v_j(X_j) = v_j(\hat{X}_j) \le v_i(\hat{X}_i) \le v_i(X_i)$. This proves the first part of the lemma. If equality holds, since $j$ was picked at the iteration in consideration, we have $j < i$ (Lemma \ref{lem:bivalued-picks-least-value}). Moreover, if equality holds, since $|X^1_j| = 0$ and $|X^1_i| = 0$, we have $|X^c_j| = |X^c_i|$.
\end{enumerate}
\end{proof}

Our next Lemma shows that no good remains unallocated under $X$, the output of \Cref{algo:bivalued-yankee-swap}.
\begin{lemma}\label{lem:x-complete}
Let $X = X^c \cup X^1$ be the final allocation output by Bivalued Yankee Swap. We have $\sum_{h \in N} |X_h| = |G|$.
\end{lemma}
\begin{proof}
Assume for contradiction that this is not true. Since we cannot allocate more goods than $|G|$, we have $\sum_{h \in N}|X_h| < |G|$. 

Breaking this down, we get $\sum_{h \in N}|X^c_h|+|X^1_h| < \sum_{h \in N+0} |X^c_h|$ which gives us $\sum_{h \in N} |X^1_h| < |X^c_0|$. However, this contradicts the termination condition of \Cref{algo:bivalued-yankee-swap}, therefore $\sum_{h \in N}|X_h| = |G|$
\end{proof}
\subsection{Properties of the dominating $\Psi$-maximizing allocation $Y$}
Here, we establish an important sufficient condition
about the $\Psi$ maximizing allocation $Y = Y^c \cup Y^1$.
In the clean allocation $Y^c$, if we can improve the value of some low valued agent $j$ at the expense of some high valued agent $i$, we can rearrange the goods in $Y^1$ such that the new allocation $\Psi$-dominates $Y$.
\begin{lemma}\label{lem:y-not-psi-maximizing}
Let $Y = Y^c \cup Y^1$ be an allocation. Let $Z^c$ and $Z^1$ be clean and supplementary allocations respectively such that
\begin{enumerate}[(a)]
    \item For all agents $h \in N$, we have $|Z^1_h| = |Y^1_h|$.
    \item For some agent $i \in N$, we have $|Z^c_i| = |Y^c_i| - 1$.
    \item For some agent $j \in N$, we have $|Z^c_j| = |Y^c_j| + 1$.
    \item For all agents $h \in N - i - j$, $|Z^c_h| = |Y^c_h|$
    \item $|Y^c_j| \le |Z^c_i|$. If equality holds, $j < i$.
\end{enumerate}
If such a $Z^c$ and $Z^1$ exist, then there exists an allocation $\hat Z$ such that either $\hat Z \succ_{\Psi} Y$ or $\hat Z =_{\Psi} Y$ and $\hat Z$ dominates $Y$.
\end{lemma}
\begin{figure}
\centering
\begin{subfigure}[t]{0.1\textwidth}
\begin{tabular}{c}
Agent\\
1\\
2\\
3
\end{tabular}
\end{subfigure}
\begin{subfigure}[t]{0.35\textwidth}
\centering
\begin{tabular}{cc}
$Y^c$ & $Y^1$\\
$\{\triangle,\Square,\hexagon,\text{\FiveFlowerOpen}\}$& $\{\times,\times,\times\}$\\
$\{\Diamond\}$& $\{\times,\times\}$\\
$\{\text{\Snowflake},\pentagon\}$& $\{\times,\times,\times\}$
\end{tabular}
\caption{The allocation $Y$}
\label{fig:ex-allocation-Y}
\end{subfigure}
\begin{subfigure}[t]{0.35\textwidth}
\centering
\begin{tabular}{cc}
$Z^c$ & $Z^1$\\
$\{\text{\FiveFlowerOpen},\pentagon,\triangle\}$& $\{\times,\times,\times\}$\\
$\{\Diamond,\text{\Snowflake}\}$& $\{\times,\times\}$\\
$\{\hexagon,\Square\}$& $\{\times,\times,\times\}$
\end{tabular}
\caption{The allocation $Z$}
\label{fig:ex-allocation-Z}
\end{subfigure}
\caption{Illustration of the allocations $Y$ and $Z$ in \Cref{lem:y-not-psi-maximizing}. All non $\times$ items are worth $c$ to the agents, and the $\times$ items are worth $1$. Agent 1 has one less item of value $c$ under $Z$, and agent 2 has one more $c$ valued item. Agent 3's utility is unchanged.}
\end{figure}
\begin{proof}
Let $i$ and $j$ be the two agents for which such an allocation $Z$ exists. Throughout the proof, we construct new allocations based on $Y$ and $Z$, but only change the items assigned to (and thus, utilities of) agents $i$ and $j$, and ignore any agent $h \ne i,j$. 
We divide the proof into cases.
\begin{enumerate}[label={\bfseries Case \arabic*:},itemindent=*,leftmargin=0cm]
\item  $|Y^c_j| < |Z^c_i|$.
There are two further subcases to consider here: 
\begin{enumerate}[label={\bfseries Sub-Case (\alph*):},itemindent=*,leftmargin=0cm]
\item $|Y^1_j| < c$.

Let $Z = Z^c \cup Z^1$. Let $\hat Z$ be the allocation $Y$ with one good removed from $Y^c_i$. 
Since $|Z_i^c| = |Y_i^c| - 1$, we have that $v_i(\hat Z_i^c) = v_i(Z_i^c) = c|Z_i^c|$. 
Furthermore, since $\hat Z_k = Y_k$ for all $k \ne i$, we have that $v_j(\hat Z_j) = v_j(Y_j)$. 
Since $|Y_j^1| < c$, we have $v_j(Y_j) = c|Y_j^c| + |Y_j^1| < c|Y_j^c|+c = c\times(|Y_j^c|+1)$. Since $|Y_j^c|< |Z_i^c|$, we have that 
$$v_j(Y_j) < c\times (|Y_j|+1) \le c|Z_i^c| =c|\hat Z_i^c| =v_i(\hat Z_i^c)\le v_i(\hat Z_i).$$  
Therefore:
\begin{align}
     v_j(\hat{Z}_j) = v_j(Y_j) <  v_i(\hat{Z}_i). \label{eq:sub-case-a-1}
\end{align}
Let $W = W^c \cup W^1$ be the allocation that results from starting at $\hat Z$ and adding an item of value $c$ to $j$\footnote{The allocation $W$ is hypothetical. The item of value $c$ need not exist. We are only intersted in manipulating the utility vector of the allocation $\hat Z$ to compare it with $Z$.}. 
Using \Cref{eq:sub-case-a-1} and (G3) (assigning additional utility to lower utility agents is preferable), we get that $\phi(\hat Z, j, c) > \phi(\hat Z, i, c)$. 
This implies, from (G1) (agreement with $\Psi$), that $W \succ_{\Psi} Y$.
In order to show that $Z \succ_{\psi} Y$, we show that $Z$ (weakly) Pareto dominates $W$. We have: 
\begin{enumerate}[(a)]
    \item $v_i(Z_i) \ge c|Z^c_i| + |Z^1_i| = c(|Y^c_i|-1) + |Y^1_i| = v_i(\hat{Z}_i) = v_i(W_i)$, and
    \item $v_j(Z_j) \ge c|Z^c_j| + |Z^1_j| = c(|Y^c_j|+1) + |Y^1_j| = v_j(\hat{Z}_j) + c = v_j(W_j)$. 
\end{enumerate}
Using (C1) (symmetric Pareto dominance), $Z \succeq_{\Psi} W$. 
Combining this with the observation that $W \succ_{\Psi} Y$, we get $Z \succ_{\Psi} Y$.



\item $|Y^1_j| \ge c$.
In this case, let $S^1_j$ be an arbitrary $c$-sized subset of $Z^1_j$. 
We rearrange the goods in $Z^1$ to construct $\hat{Z}^1$ as follows:
\begin{align*}
    \hat{Z}^1_h = 
    \begin{cases}
        Z^1_h & (h \in N - j - i) \\
        Z^1_j \setminus S^1_j & (h = j) \\
        Z^1_i \cup S^1_j & (h = i)
    \end{cases}
    .
\end{align*}
Let $\hat{Z} = Z^c \cup \hat{Z}^1$. We show that $\hat Z$ (weakly) Pareto dominates $Y$. We have 
\begin{enumerate}[(a)]
    \item $v_i(\hat{Z}_i) \ge c|Z^c_i| + |Z^1_i| + c = c|Y^c_i| + |Y^1_i| = v_i(Y_i)$, and
    \item $v_j(\hat{Z}_j) \ge c|Z^c_j| + |Z^1_j| - c = c|Y^c_j| + |Y^1_j| = v_j(Y_j)$.
\end{enumerate}

This shows that $\hat Z$ (weakly) Pareto dominates $\hat Y$. This implies $\hat Z \succeq_{\Psi} Y$. Furthermore, since $|Y^c_j| < |\hat Z^c_i|$, we have that $\vec s^{\hat {Z}^c}$ is lexicographically greater than $\vec s^{Y^c}$. This implies that $\hat Z$ dominates $Y$.
\end{enumerate}
\item $|Y^c_j| = |Z^c_i|$. 
From the statement of the lemma, we have $j < i$, By the definitions of $Y$ and $Z$, we have $|Y^c_i| = |Z^c_j|$. Now construct a supplementary allocation $\hat{Z}^1$ from $Z^1$ by swapping out $i$ and $j$'s supplementary bundles. That is:
\begin{align*}
    \hat{Z}^1_h = 
    \begin{cases}
        Z^1_h & (h \in N - j - i) \\
        Z^1_i & (h = j) \\
        Z^1_j & (h = i)
    \end{cases}
    .
\end{align*}
We show that $\hat Z = Z^c \cup \hat{Z}^1$ (weakly) symmetrically Pareto dominates $Y$. 
We have 
\begin{enumerate}[(a)]
    \item $v_i(\hat{Z}_i) \ge c|Z^c_i| + |Z^1_j|= c|Y^c_j| + |Y^1_j| = v_j(Y_j)$, and 
    \item $v_j(\hat{Z}_j) \ge c|Z^c_j| + |Z^1_i| = c|Y^c_i| + |Y^1_i| = v_i(Y_i)$.
\end{enumerate}
Using (C1) (symmetric Pareto dominance), we conclude that $\hat Z \succeq_{\Psi} Y$.

If the inequality is strict, we are done. 
If equality holds, then equality holds in (a) and (b) above as well. 
In this case, since $|\hat{Z}^c_i| = |Y^c_j|$ and $|\hat{Z}^c_j| = |Y^c_i|$, we get that $\vec s^{\hat{Z}^c} = \vec s^{Y^c}$. However, since $j < i$, we have that $\vec u^{\hat{Z}^c} \succ_{lex} \vec u^{Y^c}$. Therefore, $\hat Z$ dominates $Y$.
\end{enumerate}
\end{proof}
\subsection{Relating the Output of Bivalued Yankee Swap and the Dominating $\Psi$-Maximizing Allocation}
We establish some important connections between the output of the Bivalued Yankee Swap algorithm and the dominating $\Psi$ maximizing allocation. Our first Lemma shows that if an agent $i$ is removed from $U$ at any point in the algorithm, it has a weakly higher number of high valued goods compared to the optimal allocation $Y$.
\begin{lemma}\label{lem:xc-less-yc-for-agents-not-in-p}
Let $X^c$ and $X^1$ be the final allocations output by Bivalued Yankee Swap with a gain function $\phi$ and a justice criterion $\Psi$. Let $Y = Y^c \cup Y^1$ be a dominating $\Psi$-maximizing allocation.
Then for all the agents $i$ who are not in $U$ after the final iteration of bivalued Yankee Swap, $|X^c_i| \ge |Y^c_i|$.
\end{lemma}
\begin{proof}
Assume for contradiction that this is not the case. 
Let $i$ be the agent with least $|X^c_i|$ such that $|X^c_i| < |Y^c_i|$ and $i$ is not in $U$ after the final iteration of the algorithm. 
Break ties by choosing the agent with the least index $i$. 
Let $W^c$ and $W^1$ be the clean and supplementary allocations at the start of the iteration when $i$ was removed from $U$. 
Since by assumption, $i$ was not in $U$ after the final iteration of the algorithm and agents removed from $U$ are not added back, this allocation is well-defined. 
We refer to this iteration as $t$. 
From Lemma \ref{lem:xc-clean}, $W^c$ is clean. 

We have the following Lemma.
\begin{lemma}
For all agents $h \in N$, $|W^c_h| \le |Y^c_h|$.
\end{lemma}
\begin{proof}
Assume for contradiction that this is not true. Let $j \in N$ be the agent with least $|Y^c_j|$ such that $|Y^c_j| < |W^c_j|$; break ties by choosing the agent with the least index $j$.

Let us consider the iteration where $W^c_j$ received its $|W^c_j|$-th good. We refer to this iteration as $t'$. From Lemma \ref{lem:bivalued-picks-least-value}, we get that $j$ must have had the least value among all the agents in $U$. Note that $i$ was still in $U$ at iteration $t'$ since $t' < t$ and agents removed from $U$ are not added back. 
Note that ties are broken in favor of the agent with lower index; that is, if $|W^c_j| - 1 = |W^c_i|$, then $j < i$. 
This gives us the following observation.
\begin{obs}\label{obs:yj-less-wj-less-wi}
$|Y^c_j| \le |W^c_j| - 1 \le |W^c_i|$. 
If $|Y^c_j| = |W^c_i|$, then $j < i$.
\end{obs}

Using Lemma \ref{lem:augmentation-sufficient} with the allocations $W^c, Y^c$ and agent $j$, we get that there exists a path in the exchange graph $\cal G(Y^c)$ from $F_j(Y^c)$ to some agent $k \in N+0$ such that $|Y^c_k| > |W^c_k|$. We have the following two cases.

%

\begin{enumerate}[label={\bfseries Case \arabic*:},itemindent=*,leftmargin=0cm]
\item There is no path from $F_j(Y^c)$ to $Y_0$.
This implies that $k$ cannot be the dummy agent representing the unassigned items, i.e. $k \in N$.
We transfer goods along the shortest path from $F_j(Y^c)$ to $Y^c_k$ in the exchange graph $\cal G(Y^c)$, which results in a clean allocation $Z^c$ such that $|Z^c_j| = |Y^c_j| + 1$, $|Z^c_k| = |Y^c_k| -1$ and for all other $h \in N+0 -j -k$, $|Z^c_h| = |Y^c_h|$. 
By our assumption, the shortest path from $j$ to $k$ does not pass through any good in $Y_0$; otherwise, we contradict our assumption that there is no path from $F_j(Y^c)$ to $Y^c_0$. This implies, the goods in $Y^c_0$ and therefore, $Y_0$ are unaffected by the path augmentation. 
In other words, we have $Z^c_0 = Y^c_0$.
Now, consider $|W^c_k|$. 
We must have $|W^c_i| \le |W^c_k|$. 
Assume for contradiction that this is untrue. If $|W^c_i| > |W^c_k|$ and $k \in U$ at iteration $t$, we contradict Lemma \ref{lem:bivalued-picks-least-value} since we chose the iteration where $i$ was removed from $U$. 
If $k \notin U$ at iteration $t$, then $|X^c_k| = |W^c_k| < |W^c_i| \le |X^c_i|$ since once $k$ is removed from $U$, no goods get added to $X^c_k$. 
This contradicts our choice of $i$ since $|X^c_k| = |W^c_k| < |Y^c_k|$ by our choice of $k$ and $|X^c_k| < |X^c_i|$.
Furthermore, since $k \in U$ at iteration $t'$, if $|W^c_i| = |W^c_k|$ we must have $i \le k$ by our tie breaking procedure. 
Therefore, we have the following observation:

\begin{obs}\label{obs:wi-less-wk}
$|W^c_i| \le |W^c_k| \le |Z^c_k|$. If $|W^c_i| = |Z^c_k|$, $i \le k$.
\end{obs}

Combining Observations \ref{obs:yj-less-wj-less-wi} and \ref{obs:wi-less-wk}, we get that $|Y^c_j| \le |Z^c_k|$. If equality holds $j < k$. 

Applying Lemma \ref{lem:y-not-psi-maximizing} with clean and supplementary allocations $Z^c$ and $Y^1$, we contradict our assumption on $Y$ being a dominating $\Psi$ maximizing allocation.
Note that $Y^1$ can be a supplementary allocation to $Z^c$ since $Z^c_0 = Y^c_0$.

\item There is a path from $F_j(Y^c)$ to $Y_0$. 
We transfer goods along the shortest path from $F_j(Y^c)$ to $Y_0 = Y^c_0$ in the exchange graph $\cal G(Y^c)$ giving us a clean allocation $Z^c$ such that $|Z^c_j| = |Y^c_j| + 1$, $|Z^c_0| = |Y^c_0| - 1$ and for all other $h \in N+0 -j$, $|Z^c_h| = |Y^c_h|$. 
In the shortest path from $F_j(Y^c)$ to $Y_0$, a good from $Y_0$ only appears at the end of the path. This implies that only one good is removed from $Y^c_0$ to create $Z^c_0$; the rest of the bundles remain the same. In other words, $Z^c_0 = Y^c_0 - g$ for some $g$.
If the good $g$ removed from $Y^c_0$ does not belong to any $Y^1_u$ for some $u \in N$, we are done since $Z = Z^c \cup Y^1$ Pareto dominates $Y$; this implies $Z \succ_{\Psi} Y$ (C1) which contradicts our assumption on $Y$. Therefore, assume the good $g$ removed from $Y^c_0$ comes from $Y^1_u$ for some $u \in N$.
\end{enumerate}

We create new allocations $\hat Z^c$ and $\hat Z^1$ starting at $Z^c$ and $Z^1$ and moving a good from $Z^c_i$ to $\hat{Z}^1_u$.
From Observation \ref{obs:yj-less-wj-less-wi} and our initial assumption that $|W^c_i| < |Y^c_i|$, we have that $|Y^c_j| \le |\hat{Z}^c_i|$. Further, if equality holds, $j < i$. We can therefore apply Lemma \ref{lem:y-not-psi-maximizing} with $\hat{Z}^c$ and $\hat{Z}^1$ to contradict our assumption that $Y$ is a dominating $\Psi$ maximizing allocation.
\end{proof}

So far we have $|W^c_i| < |Y^c_i|$ and for all other agents $h \in N - i$, $|W^c_h| \le |Y^c_h|$ (from the above Lemma). Therefore, using Lemma \ref{lem:augmentation-sufficient}, we get that there is a path from $i$ to $0$ in $W$ --- $0$ is the only agent such that $|W^c_0| = |W_0| > |Y_0| = |Y^c_0|$. However, since we chose the iteration where $i$ was removed from $U$, this creates a contradiction.
\end{proof}

Our next lemma shows that for any agent $i \in N$, $|X^c_i|$ cannot be greater than $|Y^c_i|$.

\begin{lemma}\label{lem:yc-always-greater-than-xc}
Let $X^c$ and $X^1$ be the final allocations output by Bivalued Yankee Swap used to maximize some fairness objective $\Psi$. Let $Y = Y^c \cup Y^1$ be an allocation that maximizes $\Psi$ such that among all the allocations which maximize $\Psi$, $Y$ is dominating. For all agents $i \in N$, we have $|X^c_i| \le |Y^c_i|$.
\end{lemma}
\begin{proof}
Assume for contradiction that there exists a $j \in N$ such that $|X^c_j| > |Y^c_j|$. If there are multiple agents, choose the agent $i$ with least $|Y^c_j|$, breaking ties in favor of the lowest index agent.

Using Lemma \ref{lem:augmentation-sufficient} with the allocations $X^c, Y^c$ and agent $j$, we get that there exists a path in the exchange graph $\cal G(Y^c)$ from $F_j(Y^c)$ to some agent $k \in N+0$ such that $|Y^c_k| > |X^c_k|$. We have the following two cases.
\begin{enumerate}[label={\bfseries Case \arabic*:},itemindent=*,leftmargin=0cm]
\item There is no path from $F_j(Y^c)$ to $Y_0$. 
This implies that $k\ne 0$, i.e. $k$ is not the agent that represents the unassigned items, and $k \in N$.
We transfer goods along the path from $F_j(Y^c)$ to $Y^c_k$ to get the clean allocation $Z^c$. Note that since $|Y^c_k| > |X^c_k|$, using Lemma \ref{lem:xc-less-yc-for-agents-not-in-p} we have that $k \in U$ at the end of the algorithm. Also note that by assumption, $Y_0 = Z_0$ since the path did not contain any good in $Y_0$.  
We have two further subcases:
\begin{enumerate}[label={\bfseries Sub-Case (\alph*):},itemindent=*,leftmargin=0cm]
\item $j \notin U$ at the end of the algorithm.
Using Lemma \ref{lem:not-in-p-less-than-in-p} and the fact that $k \in U$, we have that $c|Y^c_j| < c|X^c_j| \le v_j(X_j) \le v_k(X_k) = c|X^c_k| \le c|Z^c_k|$. 

We can therefore apply Lemma \ref{lem:y-not-psi-maximizing} with allocations $Z^c$ and $Y^1$ to contradict our assumption on $Y$.

\item $j \in U$ at the end of the algorithm.
When $j$ received its final good, it was the least valued agent in $U$ (Lemma \ref{lem:bivalued-picks-least-value}). Ties were broken by choosing the agent with least index. We therefore have $|Y^c_j| \le |X^c_j| - 1 \le |X^c_k| \le |Z^c_k|$. If equality holds throughout, $j < k$.

We can again apply Lemma \ref{lem:y-not-psi-maximizing} with the allocations $Z^c$ and $Y^1$ to contradict our assumption on $Y$.
\end{enumerate}
\item There is a path from $F_j(Y^c)$ to $Y_0$.
We transfer goods along the shortest such path from $F_j(Y^c)$ to $Y^c_0$ to get the clean allocation $Z^c$. 
Note that since we transferred along the shortest path, $Z_0 = Y_0 - g$ for some $g \in G$. 

If the good $g \notin Y^1_u$ for some agent $u \in N$, then the allocation $Z = Z^c \cup Y^1$ strictly Pareto dominates the allocation $Y$. Therefore, $Z \succ_{\Psi} Y$ (C1) (symmetric Pareto dominance) and we contradict our assumption on $Y$.

Therefore, we assume that $g \in Y^1_u$ for some $u \in N$.
We have several possible subcases. 
\begin{enumerate}[label={\bfseries Sub-Case (\alph*):},itemindent=*,leftmargin=0cm]
\item $|Y^c_0| - \sum_{h \in N} |Y^1_h| > 0$.
Since by definition we have $\bigcup_{h \in N} Y^1_h \subseteq Y^c_0$, there is a good $g' \in Y^c_0 \setminus \bigcup_{h \in N} Y^1_h$. This good is not allocated to any agent, which means we can construct an allocation $Z^1$ as follows:
\begin{align*}
    Z^1_h = 
    \begin{cases}
        Y^1_h & (h \in N - u) \\
        Y^1_u - g + g' & (h = u)
    \end{cases}
    .
\end{align*}
This gives us an allocation $Z = Z^c \cup Z^1$ that strictly Pareto dominates $Y$ --- the utility of $i$ increases by at least $c$ and all the other agents have a weakly higher utility in $Z$ than in $Y$. Using (C1), we get that $Z \succ_{\Psi} Y$.

\item $j = u$.
We construct $Z^1$ as follows:
\begin{align*}
    Z^1_h = 
    \begin{cases}
        Y^1_h & (h \in N - u) \\
        Y^1_j - g & (h = j)
    \end{cases}
    .
\end{align*}
Since $|Z^c_j| = |Y^c_j| + 1$, we have $v_j(Z_j) \ge v_j(Y_j) + (c-1)$ and all the agents have a weakly higher utility in $Z$ than in $Y$. Therefore, $Z$ strictly Pareto dominates $Y$ and $Z \succ_{\Psi} Y$ (using (C1).

\item $|Y^1_j| > 0$ and $j \ne u$.
Let $g'$ be an arbitrary good in $Y^1_j$. We construct $Z^1$ as follows:
\begin{align*}
    Z^1_h = 
    \begin{cases}
        Y^1_h & (h \in N - i - u) \\
        Y^1_u - g + g' & (h = u) \\
        Y^1_j - g' & (h = j)
    \end{cases}
    .
\end{align*}
Since $|Z^c_j| = |Y^c_j| + 1$, we have $v_j(Z_j) \ge v_j(Y_j) + (c-1)$ and all the agents have a weakly higher utility in $Z$ than in $Y$. Therefore, $Z$ strictly Pareto dominates $Y$ and $Z \succ_{\Psi} Y$ (using (C1).

\item There exists a $k \in N$ such that $|Y^c_k| > |X^c_k|$.
Let $g'$ be an arbitrary good in $Z^c_k$. We create a new allocation $\hat{Z}^c$ and $\hat{Z}^1$ starting at $Z^c$ and $Z^1$ and moving the good $g'$ from $Z^c_k$ to $\hat Z^1_u$. 
This case is similar to Case 1, and implies that $Y$ is not a dominating $\Psi$-optimal allocation.

\item None of the previous four cases hold.
If none of the previous four cases hold, we can conclude that 
\begin{enumerate}[(a)]
    \item $|Y^1_j| = 0$,
    \item For all $k \in N$ we have $|Y^c_k| \le |X^c_k|$, and
    \item $|Y^c_0| - \sum_{h \in N} |Y^1_h| = 0$.
\end{enumerate}
This implies that for some agent $i \in N$, $|Y^1_i| > |X^1_i|$.
Let $g'$ be a good in $Y^1_i$. 
We define $Z^1$ as follows:
\begin{align*}
    Z^1_h = 
    \begin{cases}
        Y^1_h & (h \in N - i - u) \\
        Y^1_u - g + g' & (h = u) \\
        Y^1_i - g' & (h = i)
    \end{cases}
    .
\end{align*}
If we show that $\phi(Y, j, c) \ge \phi(Z, i, 1)$, we have that $Z \succeq_{\Psi} Y$ (Lemma \ref{lem:g5}). 
This is because all the other agents other than $i$ and $j$ have weakly greater utilities in $Z$ than in $Y$. 

Moreover, for all agents $h \in N$, we have $|Z^c_h| \ge |Y^c_h|$ with the inequality being strict for agent $j$ i.e. $|Z^c_j| > |Y^c_j|$. Therefore, $Z$ dominates $Y$.

To show $\phi(Y, j, c) \ge \phi(Z, i, 1)$, we have two further sub-cases:
\begin{enumerate}[label={\bfseries Sub-sub-Case (\roman*):},itemindent=*,leftmargin=0cm]
\item $|Y^c_i| < |X^c_i|$.
By our choice of $j$, we have $|Y^c_j| \le |Y^c_i|$ and if equality holds $j < i$. 
We have $v_j(Y_j) = c|Y^c_j| \le c|Y^c_i| = c|Z^c_i| \le v_i(Z_i)$. 

The inequality follows from our choice of $j$.

This implies $\phi(Y, j, c) = \phi(Y^c, j, c) \ge \phi(Y^c, i, c) \ge \phi(Z, i, c) > \phi(Z, i, 1) $. This first inequality uses (G3), the second inequality uses (G2), and the third inequality uses Lemma \ref{lem:g4}.

We can therefore conclude that $Z \succ_{\Psi} Y$ --- contradicting our assumption on $Y$.

\item $|Y^c_i| = |X^c_i|$.
Consider the iteration where agent $j$ received its $|X^c_j|$-th good. Let $\hat X^c$ and $\hat X^1$ be the allocation at the start of this iteration.

If agent $i$ was in $U$ at the start of this iteration, we have $|Y^c_j| \le |\hat X^c_j| \le |\hat X^c_i| \le |X^c_i| = |Y^c_i|$ from Lemma \ref{lem:bivalued-picks-least-value}. This means we can use the same analysis from Sub-sub-case (i) to show a contradiction about $Y$. 
\end{enumerate}
\end{enumerate}
\end{enumerate}
We can therefore assume $i \notin U$ at the iteration in consideration. By the choice of our iteration, we have $\phi(Y, j, c) \ge \phi(\hat X, j, c) \ge \phi(\hat X, i, 1) \ge \phi(Z, i, 1)$.
The first inequality comes from the fact that $v_j(Y_j) = c|Y^c_j| \le c|\hat X^c_j| \le v_j(X_j)$ and (G2). The second inequality comes from our choice of iteration. The third inequality comes the fact that $v_i(\hat X_i) \le v_i(X_i) \le v_i(Z_i)$ and (G2).
This gives us $Z \succeq_{\Psi} Y$ (using Lemma \ref{lem:g5}) implying that $Z$ maximizes $\Psi$ as well.

Moreover, as discussed above, $Z$ dominates $Y$ --- contradicting our assumption on $Y$.
\end{proof}

\subsection{Putting it All Together}
We are now ready to show our main result.

\thmbivaluedyankeeswap*
\begin{proof}
Let $X^c, X^1$ be the output of Bivalued Yankee Swap and $Y = Y^c \cup Y^1$ be a dominating $\Psi$ maximizing allocation. 

Assume for all $h \in N$, $v_h(X_h) \ge v_h(Y_h)$; if any inequality is strict, we have $X \succ_{\Psi} Y$ via (C1), a contradiction.

This implies for all $h \in N$, $v_h(X_h) = v_h(Y_h)$. Combining Lemmas \ref{lem:xc-less-yc-for-agents-not-in-p} and \ref{lem:yc-always-greater-than-xc}, we have that for all agents $h$ not in $U$ at the end of the algorithm, we have $|X^c_h| = |Y^c_h|$. 
Since all agents have the same utilities in both allocations, we have $|X^1_h| \le |Y^1_h|$. 
Therefore, for each agent $h \notin U$ at the end of the algorithm, we have $|X_h| \le |Y_h|$.

If for any $i$ that is in $U$ at the end of the algorithm, we have $|X_i| > |Y_i|$, then we have $v_i(X_i) = c|X_i| > c|Y_i| \ge v_i(Y_i)$ contradicting our initial assumption. 
Therefore, for all $i \in U$ at the end of the algorithm, we have $|X_i| \le |Y_i|$. 
Furthermore, for all $i \in U$ at the end of the algorithm, we have via Lemma \ref{lem:yc-always-greater-than-xc} that $|X_i| = |X^c_i| \le |Y^c_i| \le |Y_i|$. 
If any of these inequalities are strict, we have $|G| = \sum_{h \in N} |X_h| < \sum_{h \in N} |Y_h| \le |G|$, a contradiction. 
Here, the first equality holds due to Lemma \ref{lem:x-complete}. 

Therefore, equality holds throughout and we have for all agents $h$ in $N$ at the end of the algorithm $|X^c_h| = |Y^c_h|$ and $|X_h| = |Y_h|$. 
Therefore, for all agents $h \in N$, we have $|X^c_h| = |Y^c_h|$ and $v_h(X_h) = v_h(Y_h)$.
This implies $X$ maximizes $\Psi$ and among all $\Psi$ maximizing allocations $X$ is dominating. 

We, therefore, assume for contradiction that there exists an agent $i \in N$ such that $v_i(X_i) < v_i(Y_i)$. 
If there are multiple, choose the one with least $v_i(X_i)$, breaking ties in favor of agents with the smallest $|X^c_i|$, and further breaking ties in favor of agents with the smallest index. 

We have two possible cases:
\begin{enumerate}[label={\bfseries Case \arabic*:},itemindent=*,leftmargin=0cm]
\item $i \notin U$ at the end of the algorithm.

From Lemmas \ref{lem:xc-less-yc-for-agents-not-in-p} and \ref{lem:yc-always-greater-than-xc}, we get that $|Y^c_i| = |X^c_i|$. Since $v_i(X_i) < v_i(Y_i)$, we can conclude that $|X^1_i| < |Y^1_i|$.

Since $X$ allocates all the goods (Lemma \ref{lem:x-complete}), we must have that for some agent $j \in N$, $|Y_j| < |X_j|$. 

If $j \in U$ at the end of the algorithm, we have $|Y^c_j| \le |Y_j| < |X_j| = |X^c_j|$ contradicting Lemma \ref{lem:yc-always-greater-than-xc}. Therefore $j \notin U$ at the end of the algorithm.

Using Lemmas \ref{lem:xc-less-yc-for-agents-not-in-p} and \ref{lem:yc-always-greater-than-xc} again, we have that $|Y^c_j| = |X^c_j|$. Therefore $|X^1_j| > |Y^1_j|$. Let $Z^1$ be an allocation starting at $Y^1$ and moving an arbitrary good from $Y^1_i$ to $Z^1_j$.
To show that $Z \succeq_{\Psi} Y$, it suffices to show that $\phi(Y, j, 1) \ge \phi(Z, i, 1)$ (Lemma \ref{lem:g5}).

Let $W^c, W^1$ be the allocation at the iteration where $j$ received its final (provisionally allocated) good. If $i \in U$ at this iteration, we have $\phi(Y, j, 1) \ge \phi(W, j, 1) > \phi(W, i, c) \ge \phi(Z, i, c) > \phi(Z, i, 1)$. 
The first inequality comes form (G2), the second inequality comes from our choice of the iteration and the final inequality comes from Lemma \ref{lem:g4}.
Therefore, using (G1), we have $Z \succ_{\Psi} Y$ --- contradicting our assumption on $Y$. 

If $i \notin U$ at the iteration in consideration, by our choice of iteration (using Lemma \ref{lem:bivalued-picks-least-value}), we get that $v_j(Y_j) \le v_j(W_j) \le v_i(W_i) \le v_i(X_i) \le v_i(Z_i)$; if equality holds, $j < i$. This gives us $\phi(Y, j, 1) \ge \phi(W, j, 1) \ge \phi(W, i, 1) \ge \phi(Z, i, 1)$; again, if equality holds, then $j < i$. The first inequality holds due to (G2), the second inequality holds due to our choice of iteration and the final inequality holds due to (G2) again.

This gives us two possible cases. If $\phi(Y, j, 1) > \phi(Z, i, 1)$, we have $Z \succ_{\Psi} Y$ contradicting our assumption on $Y$.
If $\phi(Y, j, 1) = \phi(Z, i, 1)$, we have $v_j(Y_j) = v_i(Z_i)$ (using (G3)) and $j < i$. Therefore, both $Y$ and $Z$ have the same sorted utility vector; similarly, $Y^c$ and $Z^c$ have the same utility vector. However, since $j < i$, $Z$ dominates $Y$. Combining this with the fact that $Z =_{\Psi} Y$ --- we obtain a contradiction on $Y$.

\item $i \in U$ at the end of the algorithm.
We make two observations. First, if $i \in U$ and $v_i(X_i) < v_i(Y_i)$, we have $|X_i| = |X^c_i| \le |Y^c_i| \le |Y_i|$. If equality holds throughout, then $v_i(X_i) = v_i(Y_i)$ --- contradicting our initial assumption. Therefore, either $|X^c_i| < |Y^c_i|$ or $|X^c_i| = |Y^c_i|$ and $|X^1_i| < |Y^1_i|$.  We show a contradiction for the case where $|X^c_i| < |Y^c_i|$ --- the other case can be shown similarly since $\phi(X, i, c) > \phi(X, i, 1)$ (Lemma \ref{lem:g4}).

If $i \in U$ at the end of the algorithm, using our choice of $i$ and Lemma \ref{lem:not-in-p-less-than-in-p}, we get that for all agents $j \notin U$ at the end of the algorithm, we have $v_j(X_j) \ge v_j(Y_j)$. Combining this with the fact that $|X^c_j| = |Y^c_j|$ (Lemmas \ref{lem:xc-less-yc-for-agents-not-in-p} and \ref{lem:yc-always-greater-than-xc}), we have that $|X^1_j| \ge |Y^1_j|$.   

If all these weak inequalities are equalities, we have for all agents $j \notin U$ at the end of the algorithm, $|X_j| = |Y_j|$. Using Lemma \ref{lem:yc-always-greater-than-xc}, we get that for all agents $h$ in $U$ at the end of the allocation, $|X_h| = |X^c_h| \le |Y^c_h| \le |Y_h|$. This inequality is strict for the agent $i$. This gives us the inequality $|G| = \sum_{h \in N} |X_h| < \sum_{h \in N} |Y_h| \le |G|$ --- a contradiction.
\end{enumerate}
Therefore, for at least one agent $j \notin U$ at the end of the algorithm, $|X^1_j| > |Y^1_j|$. 
Let $Z = Z^c \cup Z^1$ be the allocation that results from starting at $Y$ and moving a good from $Y^c_i$ to $Z^1_j$.
Let $W^c, W^1$ be the allocation at the start of the iteration where $j$ receives its final (provisionally allocated) good. 

We have by our choice of the iteration, $\phi(Y, i, 1) \ge \phi(W, i, 1) > \phi(W, j, c) \ge \phi(Z, j, c)$. The first inequality holds due to (G2), the second inequality holds due to our choice of iteration and the third inequality holds due to (G2) again. Therefore, we have $Z \succ_{\Psi} Y$ (using Lemma \ref{lem:g5}) --- a contradiction on our assumption on $Y$.
\end{proof}

\section{Applying Bivalued Yankee Swap}
We now turn to applying Theorem \ref{thm:bivalued-yankee-swap} to well known fairness objectives. 
While we do not prove this explicitly, in all cases, the gain function $\phi$ can be trivially computed in time $O(\tau)$ (where $\tau$ is an upper bound on the time to compute $v_i(S)$ for any $i$ and any $S$). 
In Section \ref{sec:phi-speedup}, we show how the computation of $\phi$ can be sped up even further to $O(1)$ time. Some of the proofs in this Section have been relegated to Appendix \ref{apdx:applying-bivalued-yankee-swap} due to their similarity with Theorem \ref{thm:max-nash-welfare}.

\subsection{Max Nash Welfare}\label{sec:mnw}
Recall that a max Nash Welfare allocation $X$ is one that maximizes the number of agents $|P_X|$ who receive a non-zero utility, and subject to that maximizes the product $\prod_{i \in P_X} v_i(X_i)$. 

\begin{restatable}{theorem}{thmmaxnashwelfare}\label{thm:max-nash-welfare}
When $\Psi$ corresponds to the Nash welfare, Bivalued Yankee Swap run with the following gain function $\phi_\MNW$ computes a Nash welfare maximizing allocation.
\begin{align*}
    \phi_\MNW(X, i, d) = 
    \begin{cases}
        \frac{v_i(X_i) + d}{v_i(X_i)} & v_i(X_i) > 0 \\
        Md & v_i(X_i) = 0
    \end{cases}
\end{align*}
where $M$ is a very large positive number.
\end{restatable}
\begin{proof}
Formally, $X \succeq_{\MNW} Y$ if any of the following conditions hold:
\begin{enumerate}[(a)]
    \item $|P_X| > |P_Y|$, or
    \item $|P_X| = |P_Y|$ and $\prod_{i \in P_X} v_i(X_i) \ge \prod_{i \in P_Y} v_i(Y_i)$.
\end{enumerate}
It is easy to see that $\MNW$ satisfies (C1), and that $\phi_\MNW$ satisfies (G2) and (G3). 
The only property left to show is (G1). For any vector $\vec x \in \Z^n_{\ge 0}$, consider two agents $i$ and $j$ and two values $d_i, d_j \in \{1, c\}$. Let $\vec y$ be the vector that results from starting at $\vec x$ and adding $d_i$ to $x_i$ and $\vec z$ be the vector that results from starting at $\vec x$ and adding $d_j$ to $x_j$. We need to show that $\phi_\MNW(\vec x, i, d_i) \ge \phi_\MNW(\vec x, j, d_j)$ if and only if $\vec y \succeq_{\MNW} \vec z$ with equality holding if and only if $\vec y =_{\MNW} \vec z$. We divide the proof into cases.
\begin{enumerate}[label={\bfseries Case \arabic*:},itemindent=*,leftmargin=0cm]
\item $x_i = x_j = 0$. Then, $\vec y \succ_{\MNW} \vec z$ if and only if $d_i > d_j$ and this is true if and only if $\phi_\MNW(\vec x, i, d_i) > \phi_\MNW(\vec x, j, d_j)$. Note that $\vec y =_{\MNW} \vec z$ if $d_i = d_j$ and $\phi_\MNW$ reflects this as well.

\item $x_i > x_j = 0$.  In this case $\vec z \succ_{\MNW} \vec y$ and $\phi_\MNW(\vec x, i, d_i) < \phi_\MNW(\vec x, j, d_j)$.

\item $x_i, x_j > 0$. For this case, we use the fact that $\phi_\MNW(\vec x, i, d_i) = \frac{\prod_{h \in P_{\vec x}} y_h}{\prod_{h \in P_{\vec x}} x_h}$ and $\phi_\MNW(\vec x, j, d_j) = \frac{\prod_{h \in P_{\vec x}} z_h}{\prod_{h \in P_{\vec x}} x_h}$. Here $P_{\vec x}$ is the set of indices with positive values in the vector $\vec x$. Note that $P_{\vec x} = P_{\vec y} = P_{\vec z}$. 
\end{enumerate}

Using this observation, we can infer that $\vec y \succeq_{\MNW} \vec z$ if and only if $\phi_\MNW(\vec x, j, d_j) \ge \phi_\MNW(\vec x,i, d_i)$ with equality holding if and only if $\phi_\MNW(\vec x, i, d_i) = \phi_\MNW(\vec x, j, d_j)$.
\end{proof}

\subsection{Leximin}\label{sec:leximin}
Recall that a leximin allocation is one that maximizes the utility of the worst off agent, subject to that, maximizes the utility of the second worst off agent, and so on. 
\begin{restatable}{theorem}{thmleximin}\label{thm:leximin}
When $\Psi$ corresponds to the leximin fairness objective, Bivalued Yankee Swap run with the gain function $\phi_\lexmin(X, i, d) = -(c+1)v_i(X_i) + d$ computes a leximin allocation.
\end{restatable}


\subsection{$p$-mean Welfare}
Recall that the max $p$-mean welfare allocation $X$ first maximizes the number of agents who receive a positive utility $|P_X|$ and subject to that, maximizes $M_p(X) = \bigg (\frac1n \sum_{i \in P_X} v_i(X_i)^p \bigg )^{1/p}$ where $p \le 1$.
We have already shown how to compute this justice criterion for certain $p$ values: $M_0$ corresponds to Nash welfare (Section \ref{sec:mnw}) and $M_{-\infty}$ corresponds to leximin (Section \ref{sec:leximin}). 
We now show how to compute a $p$-mean welfare maximizing allocation for all the other $p$-values.

\begin{restatable}{theorem}{thmpmean}\label{thm:p-mean-welfare}
When $\Psi$ corresponds to the $p$-mean welfare objective with finite $p < 1$ and $p \ne 0$, Bivalued Yankee Swap run with the following gain function computes a $p$-mean welfare allocation.
\begin{align}
    \phi_\pWel(X, i, d) = 
    \begin{cases}
        (v_i(X_i) + d)^p - v_i(X_i)^p & p \in (0, 1)\text{ and } v_i(X_i) > 0 \\
        v_i(X_i)^p - (v_i(X_i) + d)^p & p < 0 \text{ and } v_i(X_i) > 0 \\
        Md & v_i(X_i) = 0
    \end{cases} \notag
\end{align}
where $M$ is a very large number.
\end{restatable}

Note that \Cref{thm:p-mean-welfare} does not include the case where $p = 1$. 
This is because the gain function $\phi(X, i, d) = d$ does not satisfy (G3). 
The case where $p = 1$ corresponds to the utilitarian social welfare (\USW) of an allocation. 
While we can construct a valid gain function to compute a \USW optimal allocation, we do not need to.
There exists an efficient algorithm for computing a \USW optimal allocation without using Bivalued Yankee Swap --- compute a clean utilitarian welfare maximizing allocation with respect to the binary submodular valuations $\{\beta_i\}_{i \in N}$ and allocate the remaining goods arbitrarily. 


\subsection{Speeding up Computation of $\phi$ to $O(1)$ time}\label{sec:phi-speedup}
The only queries we make to $\phi$ are with respect to the allocation $X = X^c \cup X^1$ maintained by Bivalued Yankee Swap. Using this simple observation, the gain function $\phi$ can be computed in $O(1)$ time for all the justice criterion discussed above if $v_i(X_i)$ can be computed in $O(1)$ time. 

This can be done since at every iteration of bivalued Yankee Swap, $v_i(X_i) = c|X^c_i| + |X^1_i|$ (Proposition \ref{prop:x-reduction}).  
Therefore, we can store the values of $|X^c_i|$ and $|X^1_i|$ for each $i \in N$ at zero cost to the asymptotic time complexity, and use them to compute $v_i(X_i)$ in $O(1)$ time. 

\subsection{When $c$ is not a natural number}
If $c$ is not a natural number (or equivalently, $a$ does not divide $b$), the complexity of the problem increases significantly and Bivalued Yankee Swap no longer works. 
This complexity is captured by \citet{akrami2022mnw}, who show that computing MNW allocations under {\em bivalued additive valuations} when $a \in \mathbb{N}_{\ge 1}$ and $b \in \mathbb{N}_{\ge 1}$ are coprime is NP-hard. More formally, they show that for every $a \ge 3$ and $b > a$, the problem of computing an \MNW allocation is NP-hard.

The problem of computing an \MNW allocation with an arbitrary real valued $c$ is a strictly harder problem and therefore, must be NP-hard as well. 
The exact same proof can also be used to show computing leximin allocations is NP-hard when $c$ is not a natural number.
That being said, it may be the case that for some specific non-integer values of $c$, we may still efficiently compute \MNW and leximin allocations. 
We leave this question to future work.

\section{Maximin Share Guarantees of MNW and Leximin Allocations}
We explore the maximin share guarantees of leximin and max Nash welfare allocations. 
The maximin share of an agent $i \in N$ (denoted by $\MMS_i$) is defined as the utility agent $i$ would receive if they divided the set of goods $G$ into $n$ bundles and picked the worst one (according to their preferences). More formally, $\MMS_i = \max_{X} \min_{j \in N} v_i(X_j)$. An allocation $X$ is defined as $\varepsilon$-$\MMS$ for some $\varepsilon > 0$ if for all agents $i \in N$, $v_i(X_i) \ge \varepsilon\MMS_i$ \citep{Budish2011EF1, procaccia2014fairenough}.

We prove the following two results. 

\begin{restatable}{theorem}{thmmnwmms}\label{thm:mnw-mms-submodular}
Let $c$ be an integer $\ge 2$. When agents have $\{1, c\}$-SUB valuations, then any max Nash welfare allocation $X$ is $\frac25$-\MMS.
\end{restatable}

\begin{restatable}{theorem}{thmlexmms}\label{thm:leximin-mms-submodular}
Let $c$ be an integer $\ge 2$.
When agents have $\{1, c\}$-SUB valuations, then any leximin allocation $X$ is $\frac1{c+2}$-\MMS. 
\end{restatable}

The proofs of these results require a careful counting-based argument to bound the number of high and low valued goods an agent receives. We complement these results with upper bounds on the \MMS guarantee of max Nash welfare and leximin allocations showing that our analysis is tight (Appendix \ref{sec:mms-upper-bounds}). 

The best known $\MMS$-guarantee for submodular valuations is $\frac13$-rd \citep{ghodsi2018fair}. Surprisingly, Theorem \ref{thm:mnw-mms-submodular} shows that the max Nash welfare allocation offers better \MMS guarantees, albeit for a restricted subclass of submodular valuations. 

\subsection{Proof of Theorem \ref{thm:mnw-mms-submodular}}

We present our main technique and proof for Theorem \ref{thm:mnw-mms-submodular} here. The proof for \Cref{thm:leximin-mms-submodular} is relegated to Appendix \ref{sec:leximin-mms}.

Throughout this section, unless otherwise mentioned, we assume each agent $i \in N$ has a bivalued submodular valuation where $c$ is an integer with value at least $2$. The proof is divided into several subsections; we start by establishing general guarantees for arbitrary justice criteria (Sections \ref{sec:leximin-identical}, \ref{sec:k-weak-dominance} and \ref{apdx:g4-def}) and then, we apply these guarantees to the Nash welfare and leximin objectives (Sections \ref{sec:mnw-mms} and \ref{sec:leximin-mms}).

\subsubsection{The Leximin Allocation with Identical Valuations}\label{sec:leximin-identical}
We start by studying leximin allocations where all agents in $N$ have the same valuation function $v_i$ for some $i \in N$. The main reason for this is that any bundle in this allocation gives any agent with valuation function $v_i$ at least their maximin share. In this section, we establish some basic bounds with respect to the dominating leximin allocation.

Formally, let $Y = Y^c \cup Y^1$ be the leximin allocation where all the agents in $N$ have the valuation function $v_i$ for some $i \in N$; if there are multiple such allocations, let $Y$ be the leximin allocation which dominates all other leximin allocations.

We know that, by the definition of a leximin allocation, $Y = Y^c \cup Y^1$, for all $j \in N$, $v_i(Y_j) \ge \MMS_i$. In particular, we focus on $Y_n$. We show that all the bundles in $Y^c$ have roughly the same size with $Y^c_n$ having the smallest size.

\begin{lemma}\label{lem:n-not-so-bad}
$|Y^c_n| \ge |Y^c_j| - 1$ for any $j \in [n-1]$.
\end{lemma}
\begin{proof}
Assume for contradiction that there exists a $j \in [n-1]$ such that $|Y^c_n| < |Y^c_j| - 1$. 

Since $Y^c_j$ and $Y^c_n$ are clean bundles with respect to the valuation function $v_i$, using Observation \ref{obs:beta-exchange}, we can infer that there exists a good $g \in Y^c_j \setminus Y^c_n$ such that $\Delta_{v_i}(Y^c_n, g) = c$.

\metaunderline{\textbf{Case 1: $|Y^1_n| < c$}.}
Create another allocation $\hat Y^c$ starting at $Y^c$ and moving the good $g$ from $Y^c_j$ to $Y^c_n$. 

Let $\hat Y = \hat{Y^c} \cup Y^1$. We have $v_i(Y_n) < v_i(\hat Y_n)$ and $v_i(Y_n) < c|Y^c_n| + c \le c|Y^c_j| - c \le v_i(\hat Y_j)$. This contradicts the fact that $Y$ is leximin since $s^{\hat Y}$ lexicographically dominates $s^{Y}$.

\metaunderline{\textbf{Case 2: $|Y^1_n| \ge c$.}}
In this case, we create $\hat Y^1$ starting at $Y^1$ and moving $c$ goods from $Y^1_n$ to $\hat Y^1_j$. 
Redefine $\hat Y$ as $\hat Y =\hat Y^c \cup \hat Y^1$.
Note that, comparing $Y$ and $\hat Y$, we add and remove a value of $c$ to both $j$ and $n$. This gives us $v_i(\hat Y_j) \ge v_i(Y_j)$ and $v_i(\hat Y_n) = v_i(Y_n)$. If the inequality is strict, we contradict the fact that $Y$ is a leximin allocation. Therefore, equality must hold and $\hat Y$ must be a leximin allocation as well. Furthermore $s^{\hat Y^c} \succ_{lex} s^{Y^c}$, which implies that $\hat Y$ dominates $Y$ --- again contradicting our choice of $Y$.
\end{proof}

\begin{lemma}\label{lem:n-worst-bundle}
For all $j \in [n-1]$, $|Y^c_n| \le |Y^c_j|$ and $v_i(Y_n) \le v_i(Y_j)$ for all $j \in [n-1]$.
\end{lemma}
\begin{proof}
Assume for contradiction that there exists a $j \in [n-1]$ such that $|Y^c_n| > |Y^c_j|$. 

Create $\hat Y$ starting at $Y$ and swapping $Y_j$ and $Y_n$. This allocation is still leximin for the instance where all the agents have the valuation $v_i$ and dominates $Y$ (since $j < n$) --- a contradiction of our choice of $Y$.
Similarly, if $v_i(Y_n) > v_i(Y_j)$, swapping the two bundles will create a leximin allocation which dominates $Y$.
\end{proof}

From Lemmas \ref{lem:n-worst-bundle} and \ref{lem:n-not-so-bad}, we infer that for each agent $j \in N$, either $|Y^c_n| = |Y^c_j|$ or $|Y^c_n| + 1 = |Y^c_j|$. We define $|N'|$ as the set of agents $j \in N$ such that $|Y^c_n| = |Y^c_j| + 1$. We use $|N'|$ to upper bound the size of $Y^1_n$.

\begin{lemma}\label{lem:y-1-n-upperbound}
$|Y^1_n| \le c + \max \left \{\frac{m - n|Y^c_n| - |N'| - c(n - |N'|)}{n}, 0 \right\}$.
\end{lemma}
\begin{proof}
Assume for contradiction that $|Y^1_n| > c + \max \left \{\frac{m - n|Y^c_n| - |N'| - c(n -|N'|)}{n}, 0 \right \}$. 

For each agent $j \in N \setminus N'$, we have from Lemma \ref{lem:n-worst-bundle}, $|Y^1_j| \ge |Y^1_n|$.
For all the other agents $j \in N'$, we have $|Y^c_j| = |Y^c_n| + 1$ (Lemma \ref{lem:n-not-so-bad}). Using Lemma \ref{lem:n-worst-bundle}, we get $|Y^1_j| \ge |Y^1_n| - c$.

This gives us
\begin{align*}
    m &\ge \sum_{j \in N} |Y^c_j| + |Y^1_j| \\
    &\ge \left (\sum_{j \in N \setminus N'} |Y^c_n| + |Y^1_n| \right ) + \left (\sum_{j \in N'} |Y^c_n| + 1 + |Y^1_n| - c \right ) \\
    &\ge n(|Y^c_n| + |Y^1_n|) + |N'| - c|N'| \\
    &> n|Y^c_n| + nc + \max \left \{m - n|Y^c_n| - |N'| - c(n - |N'|), 0 \right \} + |N'| - c|N'| \\
    &\ge n|Y^c_n| + nc + \left (m - n|Y^c_n| - |N'| - c(n - |N'|)\right ) + |N'| - c|N'| \\
    &= m,
\end{align*}
which is a contradiction. 
\end{proof}

We also prove a simple but useful lower bound on the sizes of each $Y_j$.
\begin{lemma}\label{lem:y-1-n-lowerbound}
If $|Y^1_n| > 0$,  then for all $j \in N$, $|Y_j| \ge |Y^c_n| + 1$. 
\end{lemma}
\begin{proof}
Fix any $j \in N$. If $|Y^c_j| > |Y^c_n|$, we are done. Otherwise, from Lemma \ref{lem:n-worst-bundle},  $|Y^c_j| = |Y^c_n|$. Using the fact that $v_i(Y_j) \ge v_i(Y_n)$ (Lemma \ref{lem:n-worst-bundle}), we can infer that $|Y^1_j| > 0$ which proves the lemma. 
\end{proof}

\subsubsection{$k$-Weak Dominance}\label{sec:k-weak-dominance}
Ideally, we would like to compare any $\Psi$ maximizing allocation with $Y$ to prove $\MMS$ guarantees for the $\Psi$ maximizing allocation.
However, it is challenging to bound the number of high valued goods an agent receives in any arbitrary $\Psi$ maximizing allocation. Therefore, we first develop a notion of {\em weak dominance} which allows us to make some simplifying assumptions.

\begin{definition}\label{def:k-weak-dominance}
    For a positive integer $k$, an allocation $X = X^c \cup X^1$ is {\em $k$-weakly dominating} if for all agents $i, j \in N$, we have:
    \begin{description}
        \item[(WD1)] If $|X^c_j| \ge k|X^c_i| + 2$, then there is no good $g \in X^c_j$ such that $\Delta_{v_i}(X^c_i, g) = c$.
        \item[(WD2)] If $|X_j| \ge k|X^c_i| + 2$, then there is no good $g \in X^1_j$ such that $\Delta_{v_i}(X^c_i, g) = c$.
    \end{description}
\end{definition}

When an allocation $X$ is $k$-weakly dominating and complete (no good is unallocated), then we obtain the following useful inequality. 

\begin{lemma}\label{lem:x-c-i-lowerbound}
Fix an agent $i \in N$ and construct $Y$ as a dominating leximin allocation where all the agents have the valuation function $v_i$. For any $k$-weakly dominating and complete allocation $X = X^c \cup X^1$, $|X^c_i| \ge \frac{|Y^c_n|}{k+1}$. Further, if $0 < |X^c_i| = \frac{|Y^c_n|}{k+1}$, then $|X^1_i| \ge |Y^1_n|$.
\end{lemma}
\begin{proof}
Assume for contradiction that $|X^c_i| < \frac{|Y^c_n|}{k+1}$. 

Using Lemma \ref{lem:n-worst-bundle}, we get that $|Y^c_n| \le |Y^c_j|$ for any $j \in [n-1]$. Therefore, $|X^c_i| < \frac{|Y^c_j|}{k+1}$ for any $j \in [n]$. By Observation \ref{obs:beta-exchange}, there are at least $k|X^c_i| + 1$ goods in each $|Y^c_j|$ for each $j \in [n]$ which have a marginal gain of $c$ when added to $X^c_i$. Using a pigeonhole argument, we get that at least one agent $k \in N - i$ has $k|X^c_i| + 2$ goods (in the allocation $X$) that have a marginal gain of $c$ when added to $X^c_i$. However, this contradicts the $k$-weak dominance of $X$. 

To prove the second part, assume $(k+1)|X^c_i| =|Y^c_n|$ and $|X^1_i| < |Y^1_n|$. Using a pigeonhole argument similar to the previous paragraph, we can show that each agent $j \in N - i$ has at least $k|X^c_i| > 0$ items in the allocation $X$ that give $i$ a marginal value of $c$. Using Lemma \ref{lem:y-1-n-lowerbound}, we get that 

\begin{align*}
    m = \sum_{h \in N} |Y^c_h| + |Y^1_h| \ge n|Y^c_n| + n + |X^1_i| \ge k|X^c_i| + |X^1_i| + (n-1)k|X^c_i| + n. 
\end{align*}
Using the pigeonhole principle, this implies that at least one agent $j \in N$ has $|X_j| \ge k|X^c_i| + 2$ items. Since we showed that all agents have at least one item (in $X$) that gives $i$ a marginal value of $c$, this violates $k$-weak dominance.
\end{proof}

\subsubsection{(Almost) Any $\Psi$-Maximizing Allocation}\label{apdx:g4-def}
We now turn to establishing some general inequalities that any $\Psi$-maximizing allocation is guaranteed to satisfy for any $\Psi$ that satisfies (C1) and admits a gain function $\phi$ that satisfies (G1)--(G3) (see Section \ref{sec:bi-ys-sufficient}). We impose one additional constraint on the gain function $\phi$ that significantly strengthens our results. 
\begin{description}
    \item[(G4)] For all $\vec x, \vec y \in \R^n$, if for any $i, j\in N$, $x_i > c y_j$, then $\phi(\vec x, i, c) < \phi(\vec y, j, 1)$.
\end{description}

For all of the following Lemmas, we assume $X = X^c \cup X^1$ is a $k$-weakly dominating, $\Psi$ maximizing allocation for some $k$ where $\Psi$ satisfies (C1) and admits a gain function $\phi$ that satisfies (G1)--(G4). We also fix an agent $i \in N$ and construct $Y = Y^c \cup Y^1$ as the dominating leximin allocation for the instance where all agents have the valuation function $v_i$.

Our first lemma shows that $X$ is complete; that is, $X$ allocates all the goods.

\begin{lemma}\label{lem:weak-x-complete}
$\sum_{j \in N} |X_j| = m$.
\end{lemma}
\begin{proof}
If there is some unallocated good, allocating it arbitrarily to any agent will improve that agent's utility by at least $1$ and Pareto dominate $X$ --- contradicting the fact that $X$ is $\Psi$ maximizing (using (C1)).
\end{proof}

We define $M \subseteq N - i$ as the set of the agents who have goods $g \in G$ under the allocation $X$ such that $\Delta_i(X^c_i, g) = c$. Recall that $|N'|$ consists of the set of agents $j \in [n-1]$ such that $|Y^c_j| = |Y^c_n| + 1$.

\begin{lemma}\label{lem:m-lowerbound}
$|M| \ge \frac{n(|Y^c_n| - |X^c_i|) + |N'|}{k|X^c_i| + 1}$. 
\end{lemma}
\begin{proof}
If $|Y^c_n| < |X^c_i|$, $\frac{n(|Y^c_n| - |X^c_i|) +  |N'|}{k|X^c_i| + 1}$ is negative and the inequality trivially holds. So we assume $|Y^c_n| \ge |X^c_i|$.

For each $j \in N \setminus N'$, $|Y^c_j| = |Y^c_n|$. Therefore, there are $|Y^c_j| - |X^c_i|$ items in $Y^c_j$ that give $i$ a marginal value of $c$. Similarly, for each $j \in N'$, $|Y^c_j| = |Y^c_n| + 1$. Therefore, there are $|Y^c_j| - |X^c_i| + 1$ items in $Y^c_j$ that give $i$ a marginal value of $c$.

Summing over all agents, there are a total of $n(|Y^c_n| - |X^c_i|) + |N'|$ items in $G$ that give $i$ a marginal value of $c$. From the $k$-weak dominance of $X$, the agents in $M$ who are allocated these items cannot have more than $k|X^c_i| + 1$ items. This gives us the inequality. 
\end{proof}

Our next Lemma uses (G4) to lower-bound $|X^1_i|$.

\begin{lemma}\label{lem:x-1-i-lowerbound}
$|X^1_i| > \frac{m- |M|(k|X^c_i| + 1) - |X^c_i| - (n- |M| -1)|X^c_i|c}{n} - 1$.
\end{lemma}
\begin{proof}
Assume for contradiction that this is not true. 

\begin{align*}
    & |X^1_i| \le  \frac{m- |M|(k|X^c_i| + 1) - |X^c_i| - (n- |M| -1)|X^c_i|c}{n} - 1 \\
    \implies& m - |X_i| \ge 1 + (n - |M| - 1)(|X^1_i| + c|X^c_i| + 1) + |M|(k|X^c_i| + 1) + |M||X^1_i|.
\end{align*}

Since all agents in $M$ can have at most $k|X^c_i| + 1$ items, this implies there are at least $1 + (n- |M| - 1)(|X^1_i| + c|X^c_i| + 1)$ goods allocated to agents in $N \setminus (M + i)$. Using the pigeonhole principle, this implies that at least one agent $u \in N \setminus M$ receives at least $|X^1_i| + c|X^c_i| + 2$ goods. 

If $|X^1_u| > 0$, we create an allocation $Z^1$ starting at $X^1$ and moving a good from $X^1_u$ to $Z^1_i$. 
Let $Z = X^c \cup Z^1$. Let us compare $X$ and $Z$. We have that $v_i(X_i) = |X^1_i| + c|X^c_i| < v_u(Z_u)$. This implies that $\phi(X, i, 1) > \phi(Z, u, 1)$ (G3) which implies $Z \succ_{\Psi} X$ (using Lemma \ref{lem:g5}) --- a contradiction.

If $|X^1_u| = 0$, we create $Z^c$ and $Z^1$ starting at $X^c$ and $X^1$ and moving a good from $X^c_u$ to $Z^1_i$. Let us again compare $X$ and $Z$. We have that $v_u(Z_u) = c[|X^1_i| + c|X^c_i| + 1] > c v_i(X_i)$. This implies that $\phi(X, i, 1) > \phi(Z, u, c)$ (G4) which implies $Z \succ_{\Psi} X$ (using Lemma \ref{lem:g5}) --- a contradiction.
\end{proof}

Our next Lemma examines the case where $|X^c_i| = 0$. 
\begin{lemma}\label{lem:x-c-i-zero}
Fix an agent $i \in N$ and let $Y$ be a dominating leximin allocation for the instance where all the agents in $N$ have the valuation function $v_i$. 
When $|X^c_i| = 0$, $|X^1_i| \ge \min\{c, |Y^1_n|\}$.
\end{lemma}
\begin{proof}
This proof requires a careful analysis of $X$ and $Y$. 
When $|X^c_i| = 0$, $|Y^c_n| = 0$ as well (Lemma \ref{lem:x-c-i-lowerbound} since $X$ is $k$-weakly dominating). Recall that $N'$ consists of the agents $j \in N$ where $|Y^c_j| = 1$.

This implies that there is a set $G'$ of $|N'|$ goods such that for every $g \in G'$, $v_i(X_i^c+g) = v_i(\emptyset+g) = v_i(g) = c$. 
In fact, it is easy to see that whenever $X^c_i = \emptyset$, adding any of the goods in $G'$ to $X_i$ results in a marginal gain of $c$. 
Under the allocation $X$, these $|N'|$ goods must be allocated such that:
\begin{obs}\label{obs:p1}
For all the agents $k \in N$ who receive at least one good in $G'$ under $X$, that good must be in $X^c_k$. 
\end{obs}
\begin{proof}
Assume for contradiction that this good $g \in G'$ is in $X^1_k$ for some $k \in N$. 
We have the following two cases.
\begin{enumerate}[label={\bfseries Case \arabic*:},itemindent=*,leftmargin=0cm]
\item $|X^1_i| = 0$. 
In this case, $v_i(X_i) = 0$. We create an allocation $Z = Z^c \cup Z^1$ starting at $X^c$ and $X^1$ and moving $g$ from $X^1_k$ to $Z^c_i$. 
It is easy to see that $Z$ symmetrically Pareto dominates $X$ which implies $Z \succ_{\Psi} X$ (C1).
\item $|X^1_i| > 0$.
We create an allocation $Z = Z^c \cup Z^1$ starting at $X^c$ and $X^1$ and 
\begin{inparaenum}
    \item moving $g$ from $X^1_k$ to $Z^c_i$, and
    \item moving any good from $X^1_i$ to $Z^1_k$.
\end{inparaenum}
It is easy to see that $Z$ symmetrically Pareto dominates $X$ again which implies $Z \succ_{\Psi} X$ (C1).
\end{enumerate}
\end{proof}
We also note that any agent who receives an item $g \in G'$ must have exactly one $c$-valued item.
\begin{obs}\label{obs:p2}
For all the agents $k$ who receive a good in $G'$ in $X$, then $|X^c_k| = 1$. 
\end{obs}
\begin{proof}
From Observation \ref{obs:p1}, we have that $|X^c_k| \ge 1$. If $|X^c_k| \ge 2$, we contradict the fact that $X$ is weakly dominating (WD1).
\end{proof}

Let the set of all the agents who receive a good from $G'$ under $X$ be $N''$.
Using Observations \ref{obs:p1} and \ref{obs:p2}, for all agents $j \in N'$, there is a corresponding agent $k \in N''$ such that $Y^c_j = X^c_k$. 
Note that $i \notin N''$.

Assume that $|X^1_i| < c$. The following two observations must hold.
\begin{obs}\label{obs:p3}
For all agents $k \in N$, if $|X^c_k| > 0$, then $|X^1_k| = 0$. 
\end{obs}
\begin{proof}
Assume for contradiction that $|X^1_k| > 0$.
Let $Z^1$ be the allocation starting at $X^1$ and moving a good from $X^1_k$ to $Z^1_i$. 
Let $Z = X^c \cup Z^1$.
Note that $v_k(Z_k) \ge c$ and $v_i(X_i) < c$. Therefore, $\phi(Z, k, 1) < \phi(X, i, 1)$ (G3) which implies $Z \succ_{\Psi} X$ (using Lemma \ref{lem:g5}), a contradiction.
\end{proof}

\begin{obs}\label{obs:p4}
For all $k \in N\setminus (N'' + i)$, $|X_k| \le |X^1_i| + 1$.
\end{obs}
\begin{proof}
This follows from very similar argumentation to Lemma \ref{lem:x-1-i-lowerbound}.
\end{proof}

We are now ready to go into the two cases of this proof. 
\begin{enumerate}[label={\bfseries Case \arabic*:},itemindent=*,leftmargin=0cm]
\item $|Y^1_n| \ge c$. 
Assume for contradiction that $|X^1_i| < c$. Let us count the number of goods in $X$ and $Y$. 

In $Y$, agents in $N'$ receive at least one good and agents in $N \setminus N'$ receive $c$ goods (Lemma \ref{lem:n-worst-bundle}). 
Since the allocation is complete (Lemma \ref{lem:x-complete}), we get
\begin{align}
    m \ge |N'| + c|N \setminus N'|. \label{eq:m-lowerbound-1}
\end{align}

Under $X$, the agents in $N''$ receive exactly one good (using Observations \ref{obs:p1}, \ref{obs:p2} and \ref{obs:p3}), agent $i$ receives at most $c-1$ goods and all other agents receive at most $c$ goods (using Observation \ref{obs:p4}). 
Since this allocation is complete as well (Lemma \ref{lem:weak-x-complete}), this gives us
\begin{align}
    m \le |N''| + c|N\setminus (N'' + i)| + (c-1). \label{eq:m-upperbound-1}
\end{align}

Since $|N'| = |N''|$, combining \eqref{eq:m-lowerbound-1} and \eqref{eq:m-upperbound-1}, we get the following 
\begin{align*}
    |N'| + c|N \setminus N'| \le |N''| + c|N\setminus (N'' + i)| + (c-1) 
    \implies 0 \le -1,
\end{align*}
a contradiction.
\item $|Y^1_n| < c$.
Let us again assume for contradiction that $|X^1_n| < |Y^1_n|$. Again, we count the goods in $X$ and $Y$.

In $Y$, the agents in $N'$ receive at least one good (Lemma \ref{lem:n-not-so-bad}). All other agents get at least $|Y^1_n|$ goods (from Lemma \ref{lem:n-worst-bundle}). This gives us 
\begin{align}
    m \ge |N'| + |Y^1_n||N\setminus N'|. \label{eq:m-lowerbound-2}
\end{align}

Coming to $X$, the agents in $N''$ get exactly one good (using Observations \ref{obs:p1}, \ref{obs:p2} and \ref{obs:p3}). The agents in $N\setminus (N'' + i)$ receive at most $|Y^1_n|$ goods (using Observation \ref{obs:p4}). Agent $i$ receives at most $|Y^1_n| - 1$ goods. This gives us 
\begin{align}
    m \le |N''| + |Y^1_n||N\setminus (N'' + i)| + |Y^1_n| - 1. \label{eq:m-upperbound-2}
\end{align}

Combining the two inequalities and using the fact that $|N'| = |N''|$, we get that $0 \le -1$, a contradiction.
\end{enumerate}
\end{proof}

\subsubsection{Max Nash Welfare Allocations}\label{sec:mnw-mms}
When $\Psi$ corresponds to the Nash welfare objective, it is easy to see that $\Psi$ satisfies (C1) and admits a gain function that satisfies (G1)---(G4) (Theorem \ref{thm:max-nash-welfare}). 
Therefore, to use the results from the previous section, we only need to show that max Nash welfare allocations are $k$-weakly dominating for some $k \in N$. In the following lemma, we show that max Nash welfare allocations are roughly $1$-weakly dominating.
Specifically, we show that for every max Nash welfare allocation, there exists a $1$-weakly dominating max Nash welfare allocation with the same utility vector.

\begin{lemma}\label{lem:mnw-1-weak-fairness}
For any max Nash welfare allocation $X = X^c \cup X^1$, there exists a $1$-weakly dominating max Nash welfare allocation $\hat X$ with the same utility vector. 
\end{lemma}
\begin{proof}
We first prove by contradiction that (WD1) holds. Assume there exist two agents $i, j \in N$ such that $|X_j^c| \ge |X_i^c| + 2$ and there exists a good $g \in X_j^c$ such that $\Delta_{v_i}(X^c_i, g) = c$. 
We construct a clean allocation $Y^c$ starting at $X^c$ and moving $g$ from $X^c_j$ to $Y^c_i$. 

\begin{enumerate}[label={\bfseries Case \arabic*:},itemindent=*,leftmargin=0cm]
\item $|X^1_i| < c$.
Let $Y = Y^c \cup X^1$. Let us compare the Nash welfare of $X$ and $Y$. If $v_i(X_i) = 0$, the Nash welfare of $Y$ is trivially higher than the Nash welfare of $X$ resulting in a contradiction. We can therefore assume $P_Y = P_X$. This gives us
\begin{align*}
    \frac{\prod_{h \in P_Y} v_h(Y_h)}{\prod_{h \in P_X} v_h(X_h)} = \frac{v_i(Y_i) v_j(Y_j)}{v_i(X_i) v_j(X_j)} \ge \frac{(v_i(X_i) + c) (v_j(X_j) - c)}{v_i(X_i) v_j(X_j)} = 1 - \frac{c^2}{v_i(X_i)v_j(X_j)} + \frac{c(v_j(X_j) - v_i(X_i))}{v_i(X_i)v_j(X_j)} > 1.
\end{align*}
The last inequality follows from the fact that $v_i(X_i) = c|X^c_i| + |X^1_i| < c|X^c_i| + c \le c|X^c_j| - c \le v_j(X_j) - c$.

Therefore, this case cannot occur.

\item $|X^1_i| \ge c$.
In this case, we construct $Y^1$ starting at $X^1$ and moving any $c$ goods from $X^1_i$ to $Y^1_j$. Let $Y = Y^c \cup Y^1$. Note that this is essentially a swap of value, we move a value of at least $c$ from $i$ to $j$ and a value of exactly $c$ from $j$ to $i$. Therefore, $v_i(Y_i) = v_i(X_i)$ and $v_j(Y_j) \ge v_j(X_j)$. This inequality cannot be strict since $X$ is a max Nash welfare allocation. Therefore, equality must hold and $Y$ must be a max Nash welfare allocation with the same utility vector as $X$. 

Note that every time there is a (WD1) violation, such a swap can be performed. However, this swap can be performed only a finite number of times since each time, the value of $\sum_{i \in N} |Y^c_i|^2$ strictly increases. This value is upper bounded by $m^2$. As we will show in this proof, this is also the only swap that we will need to perform.
\end{enumerate}
Next, we prove that (WD2) holds. Assume for contradiction that there are two agents $i, j \in N$ for which $|X_j| \ge |X^c_i| + 2$ and there is a good $g \in X^1_j$ such that $\Delta_{v_i}(X^c_i, g) = c$.

We construct allocations $Y^c$ and $Y^1$ by setting $Y_j^1 = X_j^1 - g$, $Y_i^c = X_i^c+g$ and keeping all other bundles the same. 
Let us now compare the max Nash welfare of $X$ and $Y = Y^c \cup Y^1$. 
\begin{enumerate}[label={\bfseries Case \arabic*:},itemindent=*,leftmargin=0cm]
\item $|X^1_i| > 0$.
In this case, we update $Y^c$ and $Y^1$ by moving any good from $Y^1_i$ to $Y^1_j$. So, we essentially replace the good in $X^1_j$ and strictly improve the utility of $i$. Therefore, $Y$ Pareto dominates $X$, contradicting the fact that $X$ is a max Nash welfare allocation.

\item $|X^1_i| = 0$.
In this case, we directly compare the Nash welfare of $Y$ and $X$. If $v_i(X_i) = 0$, the Nash welfare of $Y$ is greater than the Nash welfare of $X$. Otherwise,
\begin{align*}
    \frac{\prod_{h \in P_Y} v_h(Y_h)}{\prod_{h \in P_X} v_h(X_h)} = \frac{v_i(Y_i) v_j(Y_j)}{v_i(X_i) v_j(X_j)} \ge \frac{(v_i(X_i) + c) (v_j(X_j) - 1)}{v_i(X_i) v_j(X_j)} = 1 - \frac{c}{v_i(X_i)v_j(X_j)} + \frac{cv_j(X_j) - v_i(X_i)}{v_i(X_i)v_j(X_j)} > 1.
\end{align*}
The final inequality comes from the fact that $|X_j| \ge |X^c_i| + 2$. Therefore, $c v_j(X_j) \ge c|X_j| \ge c|X^c_i| + 2c = v_i(X_i) + 2c$. Again, we have a contradiction since $Y$ has a greater Nash welfare than $X$. This implies no (WD2) violation can occur. 
\end{enumerate}
\end{proof}

We can use the fact that max Nash welfare allocations are $1$-weakly dominating to show a relation between $X$ and the dominating leximin allocation $Y$. 

\begin{lemma}\label{lem:mnw-mms-relation}
Fix an agent $i \in N$ and let $Y$ be a dominating leximin allocation where all the agents have the valuation function $v_i$. Let $X$ be a $1$-weakly dominating max Nash welfare allocation. If $|Y^1_n| > c$ and $2|X^c_i| > |Y^c_n|$, then $c|Y^c_n| + |Y^1_n| \le \left (2 + \frac1{|X^c_i| + 1} \right )c|X^c_i| + |X^1_i| + 1 - |X^c_i|$. Additionally, if $|X^c_i| = 0$, then $|Y^1_n| \le |X^1_i| + c$.
\end{lemma}
\begin{proof}
Combining Lemmas \ref{lem:m-lowerbound} and \ref{lem:x-1-i-lowerbound}, we get
\begin{align*}
|X^1_i| > \frac{m - n|Y^c_n| + n|X^c_i| - |N'| - |X^c_i| - (n -1)|X^c_i|c}{n} + \frac{\left (n|Y^c_n| - n|X^c_i| + |N'|\right ) c|X^c_i|}{n(|X^c_i| + 1)} - 1.    
\end{align*}

Combining this with Lemma \ref{lem:y-1-n-upperbound}, we get
\begin{align}
    |Y^1_n| - |X^1_i| &< \frac{c|N'| - n|X^c_i| + nc|X^c_i|}{n} - \frac{\left (n|Y^c_n| - n|X^c_i| + |N'|\right ) c|X^c_i|}{n(|X^c_i| + 1)} + 1 \notag \\
    &= \frac{(2|X^c_i| + 1 - |Y^c_n|)c|X^c_i|}{|X^c_i| + 1} + \frac{c|N'|}{n(|X^c_i| + 1)} + 1 - |X^c_i| \notag\\
    &\le 2c|X^c_i| - c|Y^c_n| + c + \frac{c}{|X^c_i| + 1} - \frac{(2|X^c_i| + 1 - |Y^c_n|)c}{|X^c_i| + 1} + 1 - |X^c_i|. \label{eq:mnw-mms-relation}
\end{align}
If we assume $2|X^c_i| > |Y^c_n|$. Then
\begin{align*}
    |Y^1_n| - |X^1_i| &< 2c|X^c_i| - c|Y^c_n| + c + \frac{c}{|X^c_i| + 1} - \frac{2c}{|X^c_i| + 1} +1 - |X^c_i|\\ 
    &= \left (2 + \frac{1}{|X^c_i| + 1}\right ) c|X^c_i| - c|Y^c_n| + 1 - |X^c_i|.
\end{align*}
Re-arranging this gives us the first required inequality.

On the other hand, if we apply $|X^c_i| = 0$ to \eqref{eq:mnw-mms-relation}, we get
\begin{align*}
    |Y^1_n| - |X^1_i| < c + 1.
\end{align*}
Since all of these values are integers, we can relax the strict inequality to a weak inequality and re-write this as
\begin{align*}
    |Y^1_n| - |X^1_i| \le c. 
\end{align*}
This gives us the second required inequality.
\end{proof}

We now need to put all of this together for the final result.
\thmmnwmms*
\begin{proof}
Let $\hat X = \hat X^c \cup \hat X^1$ be a $1$-weakly dominating max Nash welfare allocation with the same utility vector of $X$. We know that one exists due to Lemma \ref{lem:mnw-1-weak-fairness}. If we show that $\hat X$ is $\frac{2}{5}$-\MMS, then since $\hat X$ and $X$ have the same utility vector, we show that $X$ is $\frac{2}{5}$-\MMS as well.

Consider any agent $i \in N$. Construct $Y = Y^c \cup Y^1$ as a dominating leximin allocation where all agents have the valuation function $v_i$. We have $v_i(Y_n) \ge \MMS_i$ by the definition of leximin. We show that $v_i(\hat{X}_i) \ge \frac25 \MMS_i$.

We have two cases to consider,
\begin{enumerate}[label={\bfseries Case \arabic*:},itemindent=*,leftmargin=0cm]
\item $|\hat{X}^c_i| > 0$. 
We divide this case into three sub-cases using Lemma \ref{lem:x-c-i-lowerbound}.

\begin{enumerate}[label={\bfseries Sub-Case (\alph*):},itemindent=*,leftmargin=0cm]
\item $2|\hat{X}^c_i| = |Y^c_n|$. 
This case is almost trivial from Lemma \ref{lem:x-c-i-lowerbound}.
\begin{align*}
    \MMS_i \le v_i(Y_n) = c|Y^c_n| + |Y^1_n| \le 2c|\hat{X}^c_i| + |\hat{X}^1_i| \le 2v_i(\hat{X}_i).
\end{align*}

\item $2|\hat{X}^c_i| \ge |Y^c_n| + 1$ and $|Y^1_n| \le c$. 
This case follows from a similar sequence of inequalities. 
\begin{align*}
    \MMS_i \le v_i(Y_n) = c|Y^c_n| + |Y^1_n| \le 2c|\hat{X}^c_i| - c + c \le 2v_i(\hat{X}_i). 
\end{align*}

\item $2|\hat{X}^c_i| \ge |Y^c_n| + 1$ and $|Y^1_n| > c$. 
This case follows from Lemma \ref{lem:mnw-mms-relation}.
\begin{align*}
    \MMS_i \le v_i(Y_n) = c|Y^c_n| + |Y^1_n| \le \left (2 + \frac1{|\hat{X}^c_i| + 1} \right )c|\hat{X}^c_i| + |\hat{X}^1_i| + 1 - |X^c_i| \le 2.5v_i(\hat{X}_i).
\end{align*}
The final inequality comes from using $|X^c_i| > 0$.
\end{enumerate}

\item $|\hat{X}^c_i| = 0$.
From Lemma \ref{lem:x-c-i-zero}, we get that $|\hat{X}^1_i| \ge \min\{c, |Y^1_n|\}$ and $|Y^c_i| = 0$.  There are two further subcases here.
\begin{enumerate}[label={\bfseries Sub-Case (\alph*):},itemindent=*,leftmargin=0cm]
\item $|Y^1_n| < c$. 
We get that $|\hat{X}^1_i| \ge |Y^1_n|$. This gives us 
\begin{align*}
    \MMS_i \le v_i(Y_n) = |Y^1_n| \le |\hat{X}^1_i| = v_i(\hat{X}_i).
\end{align*}

\item $|Y^1_n| \ge c$. 
We get that $|\hat{X}^1_i| \ge c$. Plugging this into Lemma \ref{lem:mnw-mms-relation}, we get 
\begin{align*}
\MMS_i \le v_i(Y_n) = |Y^1_n| \le c + |\hat{X}^1_i| \le 2|\hat{X}^1_i| \le 2v_i(\hat{X}_i).
\end{align*}
\end{enumerate}
\end{enumerate}
\end{proof}

\section{Envy-Freeness of MNW and leximin allocations}
Our final technical section deals with the envy-freeness of max Nash welfare and leximin allocations. An allocation $X$ is {\em envy free up to one good (EF1)} if for all agents $i, j \in N$, there exists a good $g \in X_j$ such that $v_i(X_i) \ge v_i(X_j - g)$ \citep{Budish2011EF1, Lipton2004EF1}. An allocation is {\em envy free up to any good (EFX)} if for all agents $i, j \in N$, and for all goods $g \in X_j$, we have $v_i(X_i) \ge v_i(X_j - g)$ \citep{Caragiannis2016MNW}.

Under binary submodular valuations, both the leximin and MNW allocations are known to be EFX \citep{Babaioff2021Dichotomous}. Under bivalued additive valuations, MNW allocations are known to be EFX \citep{amantidis2021mnw}. However, we now show that neither \MNW nor leximin allocations are EF1 under bivalued submodular valuations. 
The leximin allocation is not guaranteed to be EF1 even under bivalued additive valuations.

\begin{prop}
For every integer $c \ge 2$, there exists an instance under bivalued submodular valuations where no max Nash welfare allocation is EF1.
\end{prop}
\begin{proof}
We construct an instance with two agents and $6$ goods $\{g_1, g_2, \dots, g_6\}$. We define the valuations as
\begin{align*}
    & v_1(S) = c\min\{|S|, 2\} + \max\{|S|-2, 0\} \qquad  v_2(S) = c|S|
\end{align*}
In other words, Agent 1 values the first two items they receive at $c$, but any subsequent item at $1$. 
Any max Nash welfare allocation $X$ allocates two goods to agent $1$ and four goods to agent $2$. For any $g \in X_2$, we have $v_1(X_2 - g) = 2c + 1 > v_1(X_1)$. Therefore, no max Nash welfare allocation is EF1 in this instance.
\end{proof}

\begin{prop}
For any integer $c \ge 2$, there exists an instance under bivalued additive valuations where no leximin allocation is EF1.
\end{prop}
\begin{proof}
\Cref{ex:simple-leximin-maxnash} can be easily extended to show this claim holds. We have two agents and $2c+2$ goods. Agent 1's valuation is $v_1(S) = |S|$, and Agent 2's valuation is $v_2(S) = c|S|$. Any leximin allocation assigns two items to Agent $2$ and $2c$ items to Agent $1$. Since $c \ge2$, $2c \ge 4$, and Agent $2$ envies Agent 1 beyond one item.  
\end{proof}

\section{Conclusions and Future Work}
In this work, we study fair allocation under bivalued submodular valuations. 
Our insights about this class of valuation functions enable us to use path transfers to compute allocations which satisfy strong fairness and efficiency guarantees. 

As the first work to study bivalued submodular valuations, we believe our results are merely the tip of this iceberg. 
We believe that several other positive results can be shown, specifically with respect to \MMS guarantees. 
It is unknown whether \MMS allocations exist for this class of valuations. 
Prior results show that they indeed exist for the simpler classes of bivalued additive valuations \citep{ebadian2022bivaluedchores} and binary submodular valuations \citep{Barman2021MRFMaxmin}. 
Resolving this problem for bivalued submodular valuations is a natural next step. 

We also believe extending these results and the technique of path transfers beyond bivalued submodular valuations is a worthy pursuit. 
It is unlikely that we will be able to compute optimal \MNW or leximin allocations due to several known intractability results. 
However, we conjecture that it is possible to use a Yankee Swap based method to compute approximate max Nash welfare allocations for more general classes of submodular valuations. 
One specific class of interest is that of trivalued submodular valuations, e.g. the marginal gain of each item is either $0$, $1$ or $c>1$. 
We intend to explore this class in future work.
\bibliographystyle{plainnat}
\bibliography{abb,literature}

\begin{thebibliography}{28}
\providecommand{\natexlab}[1]{#1}
\providecommand{\url}[1]{\texttt{#1}}
\expandafter\ifx\csname urlstyle\endcsname\relax
  \providecommand{\doi}[1]{doi: #1}\else
  \providecommand{\doi}{doi: \begingroup \urlstyle{rm}\Url}\fi

\bibitem[Akrami et~al.(2022{\natexlab{a}})Akrami, Chaudhury, Hoefer, Mehlhorn,
  Schmalhofer, Shahkarami, Varricchio, Vermande, and van
  Wijland]{akrami2022halfintegers}
Hannaneh Akrami, Bhaskar~Ray Chaudhury, Martin Hoefer, Kurt Mehlhorn, Marco
  Schmalhofer, Golnoosh Shahkarami, Giovanna Varricchio, Quentin Vermande, and
  Ernest van Wijland.
\newblock Maximizing nash social welfare in 2-value instances: The half-integer
  case, 2022{\natexlab{a}}.
\newblock URL \url{https://arxiv.org/abs/2207.10949}.

\bibitem[Akrami et~al.(2022{\natexlab{b}})Akrami, Chaudhury, Hoefer, Mehlhorn,
  Schmalhofer, Shahkarami, Varricchio, Vermande, and van
  Wijland]{akrami2022mnw}
Hannaneh Akrami, Bhaskar~Ray Chaudhury, Martin Hoefer, Kurt Mehlhorn, Marco
  Schmalhofer, Golnoosh Shahkarami, Giovanna Varricchio, Quentin Vermande, and
  Ernest van Wijland.
\newblock Maximizing nash social welfare in 2-value instances.
\newblock In \emph{Proceedings of the 36th AAAI Conference on Artificial
  Intelligence (AAAI)}, 2022{\natexlab{b}}.

\bibitem[Amanatidis et~al.(2021)Amanatidis, Birmpas, Filos-Ratsikas, Hollender,
  and Voudouris]{amantidis2021mnw}
Georgios Amanatidis, Georgios Birmpas, Aris Filos-Ratsikas, Alexandros
  Hollender, and Alexandros~A. Voudouris.
\newblock Maximum nash welfare and other stories about efx.
\newblock \emph{Theoretical Computer Science}, 863:\penalty0 69--85, 2021.

\bibitem[Aziz et~al.(2022)Aziz, Li, Moulin, and Wu]{aziz2022fairdivsurvey}
Haris Aziz, Bo~Li, Hervé Moulin, and Xiaowei Wu.
\newblock Algorithmic fair allocation of indivisible items: A survey and new
  questions.
\newblock \emph{CoRR}, abs/2202.08713, 2022.

\bibitem[Babaioff et~al.(2021)Babaioff, Ezra, and
  Feige]{Babaioff2021Dichotomous}
Moshe Babaioff, Tomer Ezra, and Uriel Feige.
\newblock Fair and truthful mechanisms for dichotomous valuations.
\newblock In \emph{Proceedings of the 35th AAAI Conference on Artificial
  Intelligence (AAAI)}, pages 5119--5126, 2021.

\bibitem[Barman and Verma(2021)]{Barman2021MRFMaxmin}
Siddharth Barman and Paritosh Verma.
\newblock Existence and computation of maximin fair allocations under
  matroid-rank valuations.
\newblock In \emph{Proceedings of the 20th International Conference on
  Autonomous Agents and Multi-Agent Systems (AAMAS)}, pages 169--177, 2021.

\bibitem[Barman and Verma(2022)]{Barman2021TruthfulAF}
Siddharth Barman and Paritosh Verma.
\newblock Truthful and fair mechanisms for matroid-rank valuations.
\newblock In \emph{Proceedings of the 36th AAAI Conference on Artificial
  Intelligence (AAAI)}, 2022.

\bibitem[Barman et~al.(2020)Barman, Bhaskar, Krishna, and
  Sundaram]{barman2020tight}
Siddharth Barman, Umang Bhaskar, Anand Krishna, and Ranjani~G. Sundaram.
\newblock {Tight Approximation Algorithms for p-Mean Welfare Under Subadditive
  Valuations}.
\newblock In \emph{Proceedings of the 28th Annual European Symposium on
  Algorithms (ESA)}, pages 11:1--11:17, 2020.

\bibitem[Benabbou et~al.(2019)Benabbou, Chakraborty, Elkind, and
  Zick]{benabbou2019group}
Nawal Benabbou, Mithun Chakraborty, Edith Elkind, and Yair Zick.
\newblock Fairness towards groups of agents in the allocation of indivisible
  items.
\newblock In \emph{Proceedings of the 28th International Joint Conference on
  Artificial Intelligence (IJCAI)}, pages 95--101, 2019.

\bibitem[Benabbou et~al.(2021)Benabbou, Chakraborty, Igarashi, and
  Zick]{benabbou2021MRF}
Nawal Benabbou, Mithun Chakraborty, Ayumi Igarashi, and Yair Zick.
\newblock Finding fair and efficient allocations for matroid rank valuations.
\newblock \emph{{ACM} Transactions on Economics and Computation}, 9\penalty0
  (4), 2021.

\bibitem[Budish(2011)]{Budish2011EF1}
Eric Budish.
\newblock The combinatorial assignment problem: Approximate competitive
  equilibrium from equal incomes.
\newblock \emph{Journal of Political Economy}, 119\penalty0 (6):\penalty0 1061
  -- 1103, 2011.

\bibitem[Caragiannis et~al.(2016)Caragiannis, Kurokawa, Moulin, Procaccia,
  Shah, and Wang]{Caragiannis2016MNW}
Ioannis Caragiannis, David Kurokawa, Herv\'{e} Moulin, Ariel~D. Procaccia,
  Nisarg Shah, and Junxing Wang.
\newblock The unreasonable fairness of maximum nash welfare.
\newblock In \emph{Proceedings of the 17th ACM Conference on Economics and
  Computation (EC)}, page 305–322, 2016.

\bibitem[Chakraborty et~al.(2021)Chakraborty, Igarashi, Suksompong, and
  Zick]{chakraborty2021weighted}
Mithun Chakraborty, Ayumi Igarashi, Warut Suksompong, and Yair Zick.
\newblock Weighted envy-freeness in indivisible item allocation.
\newblock \emph{ACM Transactions on Economics and Computation}, 9, 2021.
\newblock ISSN 2167-8375.

\bibitem[Cousins(2021{\natexlab{a}})]{cousins2021axiomatic}
Cyrus Cousins.
\newblock An axiomatic theory of provably-fair welfare-centric machine
  learning.
\newblock In \emph{Proceedings of the 35th Annual Conference on Neural
  Information Processing Systems (NeurIPS)}, pages 16610--16621,
  2021{\natexlab{a}}.

\bibitem[Cousins(2021{\natexlab{b}})]{cousins2021bounds}
Cyrus Cousins.
\newblock \emph{Bounds and Applications of Concentration of Measure in Fair
  Machine Learning and Data Science}.
\newblock PhD thesis, Brown University, 2021{\natexlab{b}}.

\bibitem[Cousins(2022)]{cousins2022uncertainty}
Cyrus Cousins.
\newblock Uncertainty and the social planner’s problem: {W}hy sample
  complexity matters.
\newblock In \emph{Proceedings of the 2022 ACM Conference on Fairness,
  Accountability, and Transparency}, 2022.

\bibitem[Ebadian et~al.(2022)Ebadian, Peters, and
  Shah]{ebadian2022bivaluedchores}
Soroush Ebadian, Dominik Peters, and Nisarg Shah.
\newblock How to fairly allocate easy and difficult chores.
\newblock In \emph{Proceedings of the 21st International Conference on
  Autonomous Agents and Multi-Agent Systems (AAMAS)}, page 372–380, 2022.

\bibitem[Garg and Murhekar(2021)]{garg2021bivaluedadditive}
Jugal Garg and Aniket Murhekar.
\newblock Computing fair and efficient allocations with few utility values.
\newblock In \emph{Proceedings of the 14th International Symposium on
  Algorithmic Game Theory (SAGT)}, pages 345--359. Springer International
  Publishing, 2021.

\bibitem[Garg et~al.(2022)Garg, Murhekar, and Qin]{garg2022bivaluedchores}
Jugal Garg, Aniket Murhekar, and John Qin.
\newblock Fair and efficient allocations of chores under bivalued preferences.
\newblock In \emph{Proceedings of the 36th AAAI Conference on Artificial
  Intelligence (AAAI)}, 2022.

\bibitem[Ghodsi et~al.(2018)Ghodsi, HajiAghayi, Seddighin, Seddighin, and
  Yami]{ghodsi2018fair}
Mohammad Ghodsi, MohammadTaghi HajiAghayi, Masoud Seddighin, Saeed Seddighin,
  and Hadi Yami.
\newblock Fair allocation of indivisible goods: Improvements and
  generalizations.
\newblock In \emph{Proceedings of the 19th ACM Conference on Economics and
  Computation (EC)}, pages 539--556, 2018.

\bibitem[Heidari et~al.(2018)Heidari, Ferrari, Gummadi, and
  Krause]{heidari2018fairness}
Hoda Heidari, Claudio Ferrari, Krishna Gummadi, and Andreas Krause.
\newblock Fairness behind a veil of ignorance: A welfare analysis for automated
  decision making.
\newblock In \emph{Advances in Neural Information Processing Systems}, pages
  1265--1276, 2018.

\bibitem[Kurokawa et~al.(2018)Kurokawa, Procaccia, and
  Shah]{Kurokawa2018dichotomousallocation}
David Kurokawa, Ariel~D. Procaccia, and Nisarg Shah.
\newblock Leximin allocations in the real world.
\newblock \emph{ACM Transactions of Economics and Computation}, 6\penalty0
  (3–4), 2018.
\newblock ISSN 2167-8375.

\bibitem[Lipton et~al.(2004)Lipton, Markakis, Mossel, and
  Saberi]{Lipton2004EF1}
R.~J. Lipton, E.~Markakis, E.~Mossel, and A.~Saberi.
\newblock On approximately fair allocations of indivisible goods.
\newblock In \emph{Proceedings of the 5th ACM Conference on Economics and
  Computation (EC)}, page 125–131, 2004.

\bibitem[Moulin(2004)]{moulin2004fair}
Herv{\'e} Moulin.
\newblock \emph{Fair division and collective welfare}.
\newblock MIT Press, 2004.

\bibitem[Plaut and Roughgarden(2020)]{Plaut2017EFX}
Benjamin Plaut and Tim Roughgarden.
\newblock Almost envy-freeness with general valuations.
\newblock \emph{SIAM Journal on Discrete Mathematics}, 34\penalty0
  (2):\penalty0 1039--1068, 2020.

\bibitem[Procaccia and Wang(2014)]{procaccia2014fairenough}
Ariel~D. Procaccia and Junxing Wang.
\newblock Fair enough: Guaranteeing approximate maximin shares.
\newblock In \emph{Proceedings of the 15th ACM Conference on Economics and
  Computation (EC)}, pages 675--692, 2014.

\bibitem[Viswanathan and
  Zick(2023{\natexlab{a}})]{viswanathan2022generalyankee}
Vignesh Viswanathan and Yair Zick.
\newblock A general framework for fair allocation with matroid rank valuations.
\newblock In \emph{Proceedings of the 24th ACM Conference on Economics and
  Computation (EC)}, 2023{\natexlab{a}}.

\bibitem[Viswanathan and Zick(2023{\natexlab{b}})]{viswanathan2022yankee}
Vignesh Viswanathan and Yair Zick.
\newblock Yankee swap: a fast and simple fair allocation mechanism for matroid
  rank valuations.
\newblock In \emph{Proceedings of the 22nd International Conference on
  Autonomous Agents and Multi-Agent Systems (AAMAS)}, 2023{\natexlab{b}}.

\end{thebibliography}

\appendix
\newpage

\section{Applying Bivalued Yankee Swap}\label{apdx:applying-bivalued-yankee-swap}

\thmleximin*
\begin{proof}
Formally, $X \succ_{\lexmin} Y$ if and only if $\vec s^X \succ_{\lex} \vec s^Y$ (where $\vec s^X$ denotes the sorted utility vector of the allocation $X$).

It is easy to see that $\lexmin$ satisfies (C1). It is also easy to see that $\phi_\lexmin$ satisfies (G2) and (G3). 
The only property left to show is (G1). 
For any vector $\vec x \in \Z^n_{\ge 0}$, consider two agents $i$ and $j$ and two values $d_i, d_j \in \{1, c\}$. 
Let $\vec y$ be the vector that results from adding $d_i$ to $x_i$ and $\vec z$ be the vector that results from adding $d_j$ to $x_j$. We need to show that $\phi_\lexmin(\vec x, i, d_i) \ge \phi_\lexmin(\vec x, j, d_j)$ if and only if $\vec y \succeq_{\lexmin} \vec z$ with equality holding if and only if $\vec y =_{\lexmin} \vec z$.
\begin{enumerate}[label={\bfseries Case \arabic*:},itemindent=*,leftmargin=0cm]
\item  $x_i = x_j$. 
In this case, $\vec y \succeq_{\lexmin} \vec z$ if and only if $d_i \ge d_j$. 
By our definition of $\phi_\lexmin$, we have $\phi_\lexmin(\vec x, i, d_i) \ge \phi_\lexmin(\vec x, j, d_j)$ if and only if $d_i \ge d_j$ as well. 
If $\vec y =_{\lexmin} \vec z$, it must be that $d_i = d_j$. This occurs if and only if $\phi_\lexmin(\vec x, i, d_i) = \phi_\lexmin(\vec x, j, d_j)$.
\item $x_i < x_j$. 
Increasing the utility of the worse-off agent is always better and we have $\vec y \succ_{\lexmin} \vec z$. 
Since all vectors are integer valued and $d_j < (c+1)$, we have $\phi_\lexmin(\vec x, i, d_i) > \phi_\lexmin(\vec x, j, d_j)$ as well.
\end{enumerate}
\end{proof}

\thmpmean*
\begin{proof}
Let $\pWel$ denote the $p$-mean welfare objective. Formally, $X \succeq_{\pWel} Y$  if and only if one of the following conditions hold:
\begin{enumerate}[(a)]
    \item $|P_X| > |P_Y|$, or
    \item $|P_X| = |P_Y|$ and $\left (\frac1n \sum_{i \in N} v_i(X_i)^p \right )^{1/p} \ge \left (\frac1n \sum_{i \in N} v_i(Y_i)^p \right)^{1/p}$.
\end{enumerate}
It is easy to see that $\pWel$ satisfies (C1). It is also easy to see that $\phi$ satisfies (G2) and (G3). Similar to the previous results, the only property left to show is (G1). For any vector $\vec x \in \Z^n_{\ge 0}$, consider two agents $i$ and $j$ and two values $d_i, d_j \in \{1, c\}$. Let $\vec y$ be the vector that results from adding $d_i$ to $x_i$ and $\vec z$ be the vector that results from adding $d_j$ to $x_j$. We need to show that $\phi_\pWel(\vec x, i, d_i) \ge \phi_\pWel(\vec x, j, d_j)$ if and only if $\vec y \succeq_{\pWel} \vec z$ and equality holds if and only if $\vec y =_{\pWel} \vec z$.
\begin{enumerate}[label={\bfseries Case \arabic*:},itemindent=*,leftmargin=0cm]
\item $x_i \ge x_j = 0$. 
The proof for this case is the same as those of Cases 1 and 2 in Theorem \ref{thm:max-nash-welfare}.
\item $p \in (0, 1)$ and $x_i, x_j > 0$. 
By our assumption, we have $P_{\vec y} = P_{\vec z} = P_{\vec x}$. This gives us:
\begin{align*}
    \vec y \succeq_{\pWel} \vec z
    \Leftrightarrow  \sum_{i \in P_{\vec x}} y_i^p  \ge \sum_{i \in P_{\vec x}} z_i^p 
    \Leftrightarrow \sum_{i \in P_{\vec x}} y_i^p - \sum_{i \in P_{\vec x}} x_i^p \ge \sum_{i \in P_{\vec x}} z_i^p - \sum_{i \in P_X} x_i^p \\
    \Leftrightarrow (x_i + d_i)^p - x_i^p \ge (x_j + d_j)^p - x_j^p \\
    \Leftrightarrow \phi_{\pWel}(\vec x, i, d_i) \ge \phi_{\pWel}(\vec x, j, d_j).
\end{align*}
We can replace the inequality with an equality and the analysis holds.

\item $p < 0$ and $x_i, x_j > 0$. We still have $P_{\vec y} = P_{\vec z} = P_{\vec x}$. Using a similar argument to Case 2, we get:
\begin{align*}
    \vec y \succeq_{\pWel} \vec z 
    \Leftrightarrow  \sum_{i \in P_{\vec x}} y_i^p  \le \sum_{i \in P_{\vec x}} z_i^p 
    \Leftrightarrow \sum_{i \in P_{\vec x}} y_i^p - \sum_{i \in P_{\vec x}} x_i^p \le \sum_{i \in P_{\vec x}} z_i^p - \sum_{i \in P_{\vec x}} x_i^p \\
    \Leftrightarrow (x_i + d_i)^p - x_i^p \le (x_j + d_j)^p - x_j^p \\
    \Leftrightarrow \phi_{\pWel}(\vec x, i, d_i) \ge \phi_{\pWel}(\vec x, j, d_j).
\end{align*}
Again, we can replace the inequality with an equality and the analysis will still hold.
\end{enumerate}
\end{proof}

\section{\MMS Guarantees for Leximin Allocations}\label{sec:leximin-mms}

The proofs here are very similar to those in Section \ref{sec:mnw-mms}. Theorem \ref{thm:leximin} shows that $\Psi$ (corresponding to the leximin objective) satisfies (C1) and admits a gain function that satisfies (G1)---(G3). It is easy to see that the same gain function $\phi$ satisfies (G4) as well. The following Lemma shows that leximin allocations are roughly $c$-weakly dominating.

\begin{lemma}\label{lem:leximin-c-weak-fairness}
For any leximin allocation $X = X^c \cup X^1$, there exists a $c$-weakly dominating allocation $\hat X$ with the same utility vector as $X$.
\end{lemma}
\begin{proof}
This proof is very similar to that of Lemma \ref{lem:mnw-1-weak-fairness}. We divide this proof into two parts, one for each condition:

\metaunderline{\textbf{Part 1:}}
Assume for contradiction that (WD1) does not hold. There exists two agents $i, j \in N$ such that $|X_j^c| \ge c|X_i^c| + 2$ and there exists a good $g \in X_j^c$ such that $\Delta_{v_i}(X^c_i, g) = c$. We construct a clean allocation $Y^c$ starting at $X^c$ and moving $g$ from $X^c_j$ to $Y^c_i$. 

\metaunderline{\textbf{Case 1: $|X^1_i| < c$.}}
Let $Y = Y^c \cup X^1$. Let us compare $X$ and $Y$. We have that $v_i(X_i) < v_i(Y_i)$ and $v_i(X_i) = |X^1_i| + c|X^c_i| < c(|X^c_i| + 1) \le v_j(Y_j)$. This implies that $\vec s^{Y} \succ_{lex} \vec s^{X}$ --- contradicting our assumption that $X$ is leximin. 
Therefore, this case cannot occur.

\metaunderline{\textbf{Case 2: $|X^1_i| \ge c$.}}
In this case, we construct $Y^1$ starting at $X^1$ and moving any $c$ goods from $X^1_i$ to $Y^1_j$. Let $Y = Y^c \cup Y^1$. Note that this is essentially a swap of value, we move a value of at least $c$ from $i$ to $j$ and a value of exactly $c$ from $j$ to $i$. Therefore, $v_i(Y_i) = v_i(X_i)$ and $v_j(Y_j) \ge v_j(X_j)$. This inequality cannot be strict since $X$ is a leximin allocation. Therefore, equality must hold and $Y$ must be a leximin allocation with the same utility vector as $X$. 

Note that every time there is a (WD1) violation, such a swap can be performed. However, this swap can be performed only a finite number of times, since each time the value of $\sum_{i \in N} |Y^c_i|^2$ strictly increases, and this value is upper bounded by $m^2$. As we will show in the next part, (WD2) violations cannot happen, so this is the only swap that needs to be made.

\metaunderline{\textbf{Part 2:}} Assume for contradiction that (WD2) does not hold. That is, there are two agents $i, j \in N$ for which $|X_j| \ge c|X^c_i| + 2$ and there is a good $g \in X^1_j$ such that $\Delta_{v_i}(X^c_i, g) = c$.

Construct allocations $Y^c$ and $Y^1$ starting at $X^c$ and $X^1$ and moving $g$ from $Y^1_j$ to $Y^c_i$. Let us now compare $X$ and $Y = Y^c \cup Y^1$. 

\metaunderline{\textbf{Case 1: $|X^1_i| > 0$.}}
In this case, we update $Y^c$ and $Y^1$ by moving any good from $Y^1_i$ to $Y^1_j$. So, we essentially replace the good in $X^1_j$ and strictly improve the utility of $i$. Therefore, $\vec s^{Y} \succ_{lex} \vec s^{X}$ and $X$ cannot be a leximin allocation.

\metaunderline{\textbf{Case 2: $|X^1_i| = 0$.}}
In this case, we directly compare $Y$ and $X$. We have $v_i(X_i) < v_i(Y_i)$ and $v_i(X_i) = c|X^c_i| < |X_j| \le v_j(Y_j)$. Therefore, $\vec s^{Y} \succ_{lex} \vec s^X$ and $X$ cannot be a leximin allocation.
\end{proof}

Note that that fact that leximin allocations are $c$-weakly dominating prevents us from using Lemma \ref{lem:mnw-mms-relation} directly. We therefore prove alternate, slightly weaker, versions of Lemmas \ref{lem:x-1-i-lowerbound} and \ref{lem:mnw-mms-relation} for leximin allocations. 

Our notation follows from the previous section. We fix an agent $i \in N$ and we define $Y$ as a dominating leximin allocation where all agents have the valuation function $v_i$. We let $X$ be a $c$-weakly dominating leximin allocation.

\begin{lemma}\label{lem:leximin-x-1-i-lowerbound}
$|X^1_i| > \frac{m-|X^c_i| - (n-1)|X^c_i|c}{n} - 1$.
\end{lemma}
\begin{proof}
Assume for contradiction that this is not true. Then 
\begin{align*}
    |X^1_i| \le \frac{m-|X^c_i| - (n-1)|X^c_i|c}{n} - 1 \implies &(n-1)(|X^1_i| + 1) + |X^1_i| + 1 \le m - |X^c_i| - (n-1)c|X^c_i| \\
    \implies& m - |X_i| \ge 1 + (n-1)(|X^1_i| + c|X^c_i| + 1).
\end{align*}

This implies there are at least $1 + (n-1)(|X^1_i| + c|X^c_i| + 1)$ goods allocated to agents in $N - i$. Using the pigeonhole principle, this implies that at least one agent $u \in N$ receives at least $|X^1_i| + c|X^c_i| + 2$ goods. The rest of the proof follows very similarly to Lemma \ref{lem:x-1-i-lowerbound}.
\end{proof}

\begin{lemma}\label{lem:leximin-y-1-n-x-1-i-compare}
$|Y^1_n| - |X^1_i| \le c + c|X^c_i|$.
\end{lemma}
\begin{proof}
We divide this proof into two simple cases.
\begin{enumerate}[label={\bfseries Case \arabic*:},itemindent=*,leftmargin=0cm]

\item $|Y^1_n| \le c$. We trivially have $|Y^1_n| - |X^1_i| \le c$.



\item Case 1 does not occur.
Directly plugging in Lemmas \ref{lem:y-1-n-upperbound} and \ref{lem:leximin-x-1-i-lowerbound}, we get 
\begin{align*}
    |Y^1_n| - |X^1_i| &\le c + \frac{m - n|Y^c_n| - |N'| - c(n-|N'|)}{n} - \frac{m-|X^c_i| - (n-1)|X^c_i|c}{n} + 1 \\
    &\le c + c|X^c_i| - \left (\frac{n|Y^c_n| + (c-1)(n-|N'|)}{n} \right )\le c + c|X^c_i|.
\end{align*}
\end{enumerate}
\end{proof}

We can now combine these results in a similar way to Theorem \ref{thm:mnw-mms-submodular} to prove an \MMS-guarantee for leximin allocations.
\thmlexmms*
\begin{proof}
Let $X$ be an arbitrary leximin allocation.
Let $\hat X = \hat X^c \cup \hat X^1$ be a $c$-weakly dominating leximin allocation with the same utility vector of $X$ (Lemma \ref{lem:leximin-c-weak-fairness}). If we show that $\hat X$ is $\frac1{c+2}$-\MMS, then since $\hat X$ and $X$ have the same utility vector, we show that $X$ is $\frac1{c+2}$-\MMS as well.

Consider any agent $i \in N$. Construct $Y = Y^c \cup Y^1$ as a dominating leximin allocation where all agents have the valuation function $v_i$. We have $v_i(Y_n) \ge \MMS_i$ by the definition of leximin. 

We have two cases to consider,
\begin{enumerate}[label={\bfseries Case \arabic*:},itemindent=*,leftmargin=0cm]
\item $|\hat{X}^c_i| > 0$. 
We divide this case into two sub-cases using Lemma \ref{lem:x-c-i-lowerbound}.

\begin{enumerate}[label={\bfseries Sub-Case (\alph*):},itemindent=*,leftmargin=0cm]
\item $(c+1)|\hat{X}^c_i| = |Y^c_n|$. 
This case follows from Lemma \ref{lem:x-c-i-lowerbound}.
\begin{align*}
    \MMS_i \le v_i(Y_n) = c|Y^c_n| + |Y^1_n| \le c(c+1)|\hat{X}^c_i| + |\hat{X}^1_i| \le (c+1)v_i(\hat{X}_i).
\end{align*}

\item $(c+1)|\hat{X}^c_i| \ge |Y^c_n| + 1$. 
This case follows from a similar sequence of inequalities. We upper bound $|Y^1_n|$ using Lemma \ref{lem:leximin-y-1-n-x-1-i-compare}.
\begin{align*}
    \MMS_i \le v_i(Y_n) = c|Y^c_n| + |Y^1_n| \le c(c+1)|\hat{X}^c_i| - c + c + c|\hat{X}^c_i| + |\hat{X}^1_i| \le (c+2)v_i(\hat{X}_i).
\end{align*}
\end{enumerate}

\item $|\hat{X}^c_i| = 0$.
This case follows exactly the same way as Theorem \ref{thm:mnw-mms-submodular}.
From Lemma \ref{lem:x-c-i-zero}, we get that $|\hat{X}^1_i| \ge \min\{c, |Y^1_n|\}$ and $|Y^c_i| = 0$.  There are two further subcases here.
\begin{enumerate}[label={\bfseries Sub-Case (\alph*):},itemindent=*,leftmargin=0cm]
\item $|Y^1_n| < c$. 
We get that $|\hat{X}^1_i| \ge |Y^1_n|$. This gives us 
\begin{align*}
    \MMS_i \le v_i(Y_n) = |Y^1_n| \le |\hat{X}^1_i| = v_i(\hat{X}_i).
\end{align*}

\item $|Y^1_n| \ge c$. 
We get that $|\hat{X}^1_i| \ge c$. Plugging this into Lemma \ref{lem:leximin-y-1-n-x-1-i-compare}, we get 
\begin{align*}
\MMS_i \le v_i(Y_n) = |Y^1_n| \le c + |\hat{X}^1_i| \le 2|\hat{X}^1_i| \le 2v_i(\hat{X}_i).
\end{align*}
\end{enumerate}
\end{enumerate}
\end{proof}

\section{Maximin Share Guarantee Upper Bounds}\label{sec:mms-upper-bounds}
In this section, we show examples proving that the analysis in the previous section is tight.

\begin{example}
We start with the max Nash welfare allocation. Consider an instance with $n$ agents (where $n$ is odd) and a set of $n + 0.5c(n-1)$ goods. The goods are divided into three sets:
\begin{inparaenum}[(a)]
    \item $G_1$ consisting of $n$ goods, 
    \item $G_2$ consisting of $n-1$ goods, and 
    \item $G_3$ consisting of the remaining $0.5c(n-1) - (n-1)$ goods.
\end{inparaenum}

We construct agent valuations as follows:
\begin{align*}
    v_1(S) = (c - 1)\max\{|S \cap G_1|, 1\} + (c - 1)\max\{|S \cap G_2|, 1\} + |S| \\
    v_j(S) = (c-1)|S \cap G_2| + |S| && \forall j \in \{2, 3, \dots, \frac{n+1}{2}\} \\
    v_j(S) = (c-1)|S \cap G_1| + (c-1)|S \cap G_3| + |S| && \forall j \in \{\frac{n+3}{2}, \dots, n\}
\end{align*}

Agent $1$'s optimal \MMS partition occurs when the $n$ items in $G_1$ are spread out evenly, the $n-1$ items in $G_2$ are spread out evenly with one agent $j$ missing out, and the items in $G_3$ are spread out as evenly as possible with the agent $j$ receiving $c$ extra items to make up for them not receiving any item from $G_2$. This gives us 
\begin{align*}
    \MMS_1 &\ge 2c + \left \lfloor \frac{0.5c(n-1) - c - (n-1)}{n} \right \rfloor \\
    &\ge 2.5c - \frac{1.5c}{n} + \frac{1}{n} - 2.
\end{align*}

In any max Nash welfare allocation, agent $1$ receives only one item from $G_1$, with all items in $G_2$ being split evenly among the agents $\{2, 3, \dots, \frac{n+1}{2}\}$ and the remaining items in $G_1$ and $G_3$ being split among the agents $\{\frac{n+3}{2}, \dots, n\}$. Thus, agent $1$ only receives a utility of $c$ in the max Nash welfare allocation. With a large $c$ and an even larger $n$, this corresponds to roughly $\frac25$-ths of agent $1$'s maximin share.
\end{example}

\begin{example}
We now move to the leximin allocation. We construct an instance with $n$ agents and $n + cn(n-1)$ goods where agent $n$ has the valuation function $v_n(S) = c|S|$ and all other agents $j \in [n-1]$ have the valuations $v_j(S) = |S|$.
In this instance, any leximin allocation allocates $n$ goods to agent $n$ and $cn$ goods to all other agents. 
The maximin share of agent $n$ is trivially $c[1 + c(n-1)]$. When both $c$ and $n$ are large, agent $n$ receives roughly $\frac1c$-th of their maximin share. 
\end{example}
\end{document}